%% file: notes.tex
\documentclass[11pt,reqno]{amsart}
\input{header_new}

\begin{document}

\title[Change-point computation for large graphical models]{A scalable algorithm for Gaussian graphical models with change-points}\thanks{This work is partially supported by the NSF grant DMS 1513040}

\author{Yves Atchad\'e}\thanks{ Y. Atchad\'e: University of Michigan, 1085 South University, Ann Arbor,
  48109, MI, United States. {\em E-mail address:} yvesa@umich.edu}
\author{Leland Bybee}\thanks{ L. Bybee: University of Michigan, 1085 South University, Ann Arbor,
  48109, MI, United States. {\em E-mail address:} lelandb@umich.edu}

\subjclass[2010]{62F15, 62Jxx}

\keywords{Change-points, Gaussian graphical models, proximal gradient, simulated Annealing, Stochastic Optimization}

\maketitle
\begin{center} (\today) \end{center}

\begin{abstract}
Graphical models with change-points are computationally challenging to fit, particularly in cases where the number of observation points and the number of nodes in the graph are large.  Focusing on Gaussian graphical models, we introduce an approximate majorize-minimize (MM) algorithm that can be useful for computing change-points in large graphical models. The proposed algorithm is an order of magnitude faster than a brute force search. Under some regularity conditions on the data generating process, we show that with high probability, the algorithm converges to a value that is within statistical error of the true change-point. A fast implementation of the algorithm using Markov Chain Monte Carlo is also introduced. The performances of the
proposed algorithms are evaluated on synthetic data sets and the algorithm is also used to analyze structural  changes in the S\&P 500 over the period 2000-2016. 
\end{abstract}

\setcounter{secnumdepth}{3}

\section{Introduction}\label{sec:intro}
Networks are fundamental structures that are commonly used to describe interactions between sets of actors or nodes. In many applications, the behaviors of the actors are observed over time and one is interested in recovering the underlying network connecting these actors. High-dimensional versions of this problem where the number of actors is large (compared to the number of time points) is of special interest. In the statistics and machine learning literature, this problem is typically framed as fitting large graphical models with sparse parameters, and significant progress has been made recently, both in terms of the statistical theory (\cite{meinshausen06,yuanetlin07,barnejeeetal08,ravikumaretal11,hastie:etal:15}), and practical algorithms (\cite{friedmanetal08,hoefling09,atchade:prox14}). 

In many problems arising in areas such as biology, finance, and political sciences, it is well-accepted that the underlying networks of interest are not static, but can undergo abrupt changes over time. Graphical models with change-points (or piecewise constant graphical models) are simple, yet powerful models that are particularly well-suited for such problems. 
However, despite their conceptual simplicity, these  models are computationally challenging to fit. For instance a full grid search approach to locate a single change-point is a Gaussian graphical model with a \textsf{lasso} penalty (\textsf{glasso}) requires solving $O(T)$ \textsf{glasso} sub-problems, where $T$ is the number of time points. Most algorithms for the \textsf{glasso} problem scale  like $O(p^3)$ or worst\footnote{Furthermore the constant in the big-O is typically problem dependent and can be large}, where $p$ is the number of nodes. Hence when $p$ and $T$ are large, fitting a high-dimensional Gaussian graphical model with a single change-point has a taxing computational cost that scales at least as $O(Tp^3)$.

The literature addressing the computational aspects of change-point models is rather sparse. A large portion of change-point detection procedures are based on cumulative sums (CUSUM) or similar statistic monitoring approaches (\cite{levy:leduc:09,chen:zhang:15,cho:fritz:15} and the references therein). For simple enough statistics, these  change-point detection procedures can be efficiently implemented, and  the computational difficulty aforementioned can be avoided. However in problems where one wishes to detect structural changes in large networks, a CUSUM-based or a statistic-based approach can be difficult to employ, since it requires knowledge of the pertinent statistics to monitor. Furthermore the estimation of the change-point as well as the network structure before and after the change-point can provide new insight in the underlying phenomenon driving the changes. Hence CUSUM-based approaches may not be appropriate in applications where the main driving forces of the network changes are poorly understood, and/or are of prime interest.  In \cite{aue:etal:09} the author proposed a methodology to detect changes in the covariance structures of multivariate time-series. However their methodology is intractable in the high-dimensional setting considered in this paper. 

Specific works addressing computational issues in model-based change-point estimation include \cite{sandipan:14, leonardi:buhlmann:16}. In \cite{sandipan:14} the authors considered a discrete graphical model with change-point and proposed a two-steps algorithm for computation. However the success of their algorithm depends crucially on the choice of the coarse and refined grids, and there is limited insight on how to choose these. A related work is  \cite{leonardi:buhlmann:16} where the authors considered a high-dimensional linear regression model with change-points and proposed a dynamic programming approach to compute the change points. In the case of a single change-point their algorithm corresponds to the brute force (full-grid search) approach mentioned above.

In this work we propose an approximate majorize-minimize (MM) algorithm for fitting piecewise constant high-dimensional models. The algorithm can be applied more broadly. However to focus the idea we limit our discuss to Gaussian graphical models with an elastic net penalty. In this specific setting, the algorithm takes the form of a block update algorithm that alternates between a proximal gradient update of the graphical model parameters followed by a line search of the change-point. The proposed algorithm only solves for a single change-point. We extend it to multiple  change-points by binary segmentation. We study the convergence of the algorithm and show under some regularity conditions on the data generating mechanism that the algorithm is stable, and produces values in the vicinity of the true change-point (under the assumption that one such true change-point exists).

Each iteration of the proposed algorithm has a computational cost of $O(Tp^2 + p^3)$. Although this cost is one order of magnitude smaller than the $O(Tp^3)$ cost of the brute force approach, it can  still be large when $p$ and $T$ are both large. As a solution we propose a stochastic version of the algorithm where the line search performed to update the change-point is replaced by a Markov Chain Monte Carlo (MCMC)-based simulated annealing. The simulated annealing update is cheap ($O(p^2)$) and is used as a stochastic approximation of the full line search. We show by simulation that the stochastic algorithm behaves remarkably well, and as expected outperforms the deterministic algorithm is terms of computing time.

The paper is organized as follows. Section \ref{sec:method} contains a presentation of the Gaussian graphical model with change-points, followed by a details presentation of the proposed algorithms. We performed extensive numerical experiments to investigate the behavior of the proposed algorithms. We also use the algorithm to analyze structural  changes in the S\&P 500 over the period 2000-2016. The results are reported in Section \ref{sec:num:exp}. We gather some of the technical proofs in Section \ref{sec:proofs}.

We end this introduction with some notation that we shall used throughout the paper. We denote $\M_p$ the set of all symmetric elements of $\rset^{p\times p}$ equipped with its Frobenius norm $\normfro{\cdot}$ and associated inner product 
\[\pscal{A}{B}_{\textsf{F}}\eqdef \sum_{1\leq i\leq j\leq p}A_{ij}B_{ij}.\]
We denote $\M_p^+$ the subset of $\M_p$ of positive definite elements. For $0<a<A\leq +\infty$, let $\M_p^+(a,A)$ denote the subset of $\M_p^+$ of matrices $\theta$ such that $\lambda_{\textsf{min}}(\theta)\geq a$, and $\lambda_{\textsf{max}}(\theta)\leq A$, where $\lambda_{\textsf{min}}(M)$ (resp. $\lambda_{\textsf{max}}(M)$) denotes the smallest eigenvalue (resp. the largest eigenvalue) of $M$.
 
If $u\in\rset^p$, and $\theta\in\rset^{p\times p}$ for some integer $p\geq 1$, we denote $\|u\|_2$ the usual Euclidean norm of $u$, and $\|\theta\|_2$ the spectral norm (operator norm) of $\theta$.

\section{Fitting Gaussian Graphical models with a single change-point}\label{sec:method}
Let $\{X^{(t)},\;1\leq t\leq T\}$ be a sequence of  $p$-dimensional random vectors. The grid over which the change-points are searched is denoted $\mathcal{T}\eqdef\{n_0,\ldots,T-n_0\}$, for some integer $1\leq n_0<T$. We define
\[S_1(\tau)\eqdef\frac{1}{\tau}\sum_{t=1}^\tau X^{(t)}X^{(t)'}, \;\;\; S_2(\tau) \eqdef\frac{1}{T-\tau}\sum_{t=\tau+1}^T  X^{(t)}X^{(t)'},\;\;\tau\in\mathcal{T}.\]
 We define the regularization function
\[\wp(\theta) \eqdef \alpha\|\theta\|_1 + \frac{1-\alpha}{2}\normfro{\theta}^2,\;\theta\in\M_p,\]
where $\alpha\in[0,1)$ is a given constant, and $\|\theta\|_1 \eqdef\sum_{i\leq j}^p|\theta_{ij}|$. Then we define 
\[g_{1,\tau}(\theta) =\left\{\begin{array}{ll}\frac{1}{2}\frac{\tau}{T}\left[-\log\det(\theta)+\textsf{Tr}(\theta S_1(\tau))\right] & \mbox{ if } \theta\in\M_p^+,\\ +\infty & \mbox{otherwise},\end{array}\right.,\;\;\tau\in\mathcal{T},\]
where $\textsf{Tr}(A)$ (resp. $\det(A)$) denotes the trace (resp. the determinant) of $A$, and
\[g_{2,\tau}(\theta) =\left\{\begin{array}{ll}\frac{1}{2}\left(1-\frac{\tau}{T}\right)\left[-\log\det(\theta) +\textsf{Tr}(\theta S_2(\tau))\right] & \mbox{ if } \theta\in\M_p^+,\\ +\infty & \mbox{otherwise},\end{array}\right.,\;\;\tau\in\mathcal{T}.\]
For $j\in\{1,2\}$, we set
\begin{equation}\label{glasso_1}\hat\theta_{j,\tau} \eqdef \textsf{Argmin }_{\vartheta\in\M_p^+}\left[g_{j,\tau}(\vartheta) + \lambda_{j,\tau}\wp(\vartheta)\right],\end{equation}
for regularization parameters $\lambda_{1,\tau}>0,\lambda_{2,\tau}>0$, that we assume fixed throughout. We consider the problem of computing the change point estimate $\hat\tau$ defined as
\begin{equation}\label{est:2}
\hat\tau=\textsf{Argmin }_{\tau\in\mathcal{T}}\left[g_{1,\tau}(\hat\theta_{1,\tau}) + \lambda_{1,\tau}\wp(\hat\theta_{1,\tau}) + g_{2,\tau}(\hat\theta_{2,\tau}) + \lambda_{2,\tau}\wp(\hat\theta_{2,\tau})\right].
\end{equation}
If the minimization problem in (\ref{est:2}) has more than one solution, then $\hat\tau$ denotes any one of these solutions. The quantity $\hat\tau$ is the maximum likelihood estimate of a change point $\tau$ in the model which assumes that  $X^{(1)},\ldots,X^{(\tau)}$ are independent with common distribution $\textbf{N}(0,\theta_{1}^{-1})$, and $X^{(\tau + 1)},\ldots,X^{(T)}$ are independent with common distribution $\textbf{N}(0,\theta_{2}^{-1})$, for an unknown change-point $\tau$, and unknown precision matrices $\theta_{1}\neq \theta_{2}$.

The problem of computing the graphical lasso (\textsf{glasso}) estimators $\hat\theta_{j,\tau}$ in (\ref{glasso_1}) has received a lot of attention in the literature, and several efficient algorithms have been developed for this purpose (see for instance \cite{atchade:etal:15} and the references therein). Hence in principle, using any of these available \textsf{glasso} algorithms, the change-point problem in (\ref{est:2}) can be solved by solving $T-2n_0+1=O(T)$ \textsf{glasso} sub-problems. However this brute force approach can be very time-consuming in cases where $p$ and $T$ are large. For instance, one of the  most cost-efficient algorithm for solving the \textsf{glasso} problem in high-dimensional cases is the standard proximal gradient algorithm (\cite{bala:etal:12,atchade:etal:15}), which has a computational cost of $O(p^3\textsf{cond}(\hat\theta)^2\log(1/\delta))$ to deliver a $\delta$-accurate solution  (that is $\|\theta-\hat\theta\|_{\textsf{F}}\leq\delta$), where $\textsf{cond}(A)$ denotes the condition number of $A$, that is the ratio of the largest eigenvalue over the smallest eigenvalue of $A$. Hence when $p$ and $T$ are large the computational cost of the brute force approach for computing (\ref{est:2}) is of order $O\left(Tp^3\textsf{cond}(\hat\theta_{j,\tau})^2\log(1/\delta)\right)$, which can become prohibitively large. 

We propose an algorithm that we show has a better computational complexity. To motivate the algorithm we first introduce a majorize-minimize (MM) algorithm for solving (\ref{est:2}). We refer the reader to \cite{wu:lange:2010} for a general introduction to MM algorithms. Let
\[G(t)\eqdef g_{1,t}(\hat\theta_{1,t}) + \lambda_{1,t}\wp(\hat\theta_{1,t}) + g_{2,t}(\hat\theta_{2,t}) + \lambda_{2,\tau}\wp(\hat\theta_{2,t}),\;\;t\in\mathcal{T}\] 
denote the objective function of the minimization problem in (\ref{est:2}). For $\theta_1,\theta_2\in\M_p$, we also define
\begin{equation}\label{def:H}
\mathcal{H}(\tau\vert \theta_1,\theta_2) \eqdef  g_{1,\tau}(\theta_1) +\lambda_{1,\tau}\wp(\theta_1) + g_{2,\tau}(\theta_2) +\lambda_{2,\tau}\wp(\theta_2),\;\;\tau\in\mathcal{T}.\end{equation}
Instead of the brute force approach that requires solving (\ref{glasso_1}) for each value $\tau\in\mathcal{T}$, consider the following algorithm.

\begin{algorithm}[MM algorithm]\label{algo:MM}\;\; Pick $\tau^{(0)}\in\mathcal{T}$, and for $k=1,\ldots,K$, repeat the following steps.
\begin{enumerate}
\item Given $\tau^{(k-1)}\in \mathcal{T}$, compute $\hat\theta_{1,\tau^{(k-1)}}$ and $\hat\theta_{2,\tau^{(k-1)}}$, and minimize the function $\mathcal{H}(t\vert \hat\theta_{1,\tau^{(k-1)}}, \hat\theta_{2,\tau^{(k-1)}})$ to get $\tau^{(k)}$:
\[\tau^{(k)}  = \textsf{Argmin}_{t\in\mathcal{T}}\; \mathcal{H}(t\vert \hat\theta_{1,\tau^{(k-1)}}, \hat\theta_{2,\tau^{(k-1)}}).\]
\end{enumerate}
\begin{flushright}
$\square$
\end{flushright}
\end{algorithm}
\medskip

By definition of $\hat\theta_{j,\tau}$ in (\ref{glasso_1}), we have $G(t)\leq \mathcal{H}(t\vert \hat\theta_{1,\tau^{(k-1)}}, \hat\theta_{2,\tau^{(k-1)}})$ for all $t\in\mathcal{T}$. Furthermore $G(\tau^{(k-1)}) = \mathcal{H}(\tau^{(k-1)}\vert \hat\theta_{1,\tau^{(k-1)}}, \hat\theta_{2,\tau^{(k-1)}})$. Therefore, for all $k\geq 1$,
\[G(\tau^{(k)})\leq \mathcal{H}(\tau^{(k)}\vert \hat\theta_{1,\tau^{(k-1)}}, \hat\theta_{2,\tau^{(k-1)}}) \leq \mathcal{H}(\tau^{(k-1)}\vert \hat\theta_{1,\tau^{(k-1)}}, \hat\theta_{2,\tau^{(k-1)}}) = G(\tau^{(k-1)}).\]

Hence the objective function $G$ is non-increasing along the iterates of Algorithm \ref{algo:MM}.  Note that this algorithm is already potentially faster than the brute force approach, particular when $T$ is large, since we compute the graphical-lasso solutions $\hat\theta_{j,\tau^{(k)}}$ only for time points visited along the iterations. We propose to further reduce the computational cost by computing the solutions $\hat\theta_{j,\tau^{(k)}}$ only approximately.

Given $\gamma>0$, and a matrix $\theta\in\rset^{p\times p}$, define $\Prox_{\gamma}(\theta)$ (the proximal map with respect to the penalty function $\wp(\theta) = \alpha\|\theta\|_1 + (1-\alpha)\normfro{\theta}^2/2$) as the symmetric $\rset^{p\times p}$ matrix such that for $1\leq i,j\leq p$,
\[\left(\Prox_\gamma(\theta)\right)_{ij}=\left\{\begin{array}{ll} 0 & \mbox{ if } |\theta_{ij}|<  \alpha\gamma\\
\frac{\theta_{ij}-\alpha\gamma}{1+(1-\alpha)\gamma} & \mbox{ if } \theta_{ij}\geq \alpha\gamma \\
 \frac{\theta_{ij}+\alpha\gamma}{1+(1-\alpha)\gamma} & \mbox{ if } \theta_{ij}\leq  -\alpha\gamma\, .\end{array}\right.
\]

We consider the following algorithm.

\begin{algorithm}\label{algo:1}[Approximate MM algorithm]
Fix a step-size $\gamma>0$. Pick some initial value $\tau^{(0)}\in\mathcal{T}$, $\theta_1^{(0)},\theta_2^{(0)}\in\M_p^+$.  Repeat for $k=1,\ldots,K$. Given ($\tau^{(k-1)}$, $\theta_1^{(k-1)}$, $\theta^{(k-1)}_2$), do the following:
\begin{enumerate}
\item Compute
\[\theta_1^{(k)} = \Prox_{\gamma\lambda_{1,\tau^{(k-1)}}}\left(\theta_1^{(k-1)}-\gamma\left(S_1(\tau^{(k-1)})-(\theta_1^{(k-1)})^{-1}\right)\right),\]
\item compute
\[\theta_2^{(k)} = \Prox_{\gamma\lambda_{2,\tau^{(k-1)}}}\left(\theta_2^{(k-1)}-\gamma\left(S_2(\tau^{(k-1)})-(\theta_2^{(k-1)})^{-1}\right)\right),\]
\item compute 
\[\tau^{(k)} \eqdef \textsf{Argmin}_{t\in\mathcal{T}}\; \mathcal{H}\left(t\vert \theta_1^{(k)},\theta_2^{(k)}\right).\]
\end{enumerate}
\begin{flushright}
$\square$
\end{flushright}
\end{algorithm}
\medskip

Note that, if instead of a single proximal gradient update in Step (1)-(2), we do a large number proximal gradient updates (an infinite number for the sake of the argument), we recover exactly Algorithm \ref{algo:MM}. Hence Algorithm \ref{algo:1} is an approximate version of Algorithm  \ref{algo:MM}.

\begin{remark}\label{rem:1}
\begin{enumerate}
\item Notice that one can easily compute $\mathcal{H}(\tau+1\vert \theta_1,\theta_2)$ from $\mathcal{H}(\tau\vert \theta_1,\theta_2)$ by a rank-one update in $O(p^2)$ number of operations. Hence the computational cost of Step (3) is $O(Tp^2)$. And the total computational cost of one iteration of Algorithm \ref{algo:1} is $O(p^3 + Tp^2)$.
\item In practice one needs to exercise some care in choosing the step-size $\gamma$. As we show below, a small enough $\gamma$ is needed in order to maintain positive definiteness of the matrices $\theta_1^{(k)}$ and $\theta^{(k)}_2$. However, too small values of $\gamma$ lead to slow convergence. A nice trade-off that works well from the software engineering viewpoint is to start with a large value of $\gamma$ and to re-initialize the algorithm with a smaller $\gamma$ if at some point positive definiteness is lost. This issue is discussed more extensively in \cite{atchade:etal:15}.
\end{enumerate}
\end{remark}

Algorithm \ref{algo:1} raises two basic questions. The first question is whether the algorithm is stable, where here by stability we mean whether the algorithm runs without breaking down.  Indeed we notice that Steps (1 and 2) involve taking the inverse of the matrices $\theta_1^{(k-1)}$, and $\theta_2^{(k-1)}$, but there is  no guarantee a priori that these matrices are non-singular. Using results established in \cite{atchade:etal:15}, we answer this question by showing below that if the step-size $\gamma$ is small enough then the algorithm is actually stable. The second basic question is whether the algorithm converges to the optimal value. We address this question below. 

For $j\in\{1,2\}$, we set
\[\underline{\lambda}_j \eqdef \min_{\tau\in \mathcal{T}} \lambda_{j,\tau},\;\; \bar\lambda_j \eqdef \max_{\tau\in \mathcal{T}} \lambda_{j,\tau},\;\; \mu_j \eqdef \max_{\tau\in\mathcal{T}}\left[\frac{1}{2}\|S_j(\tau)\|_2 +\alpha p\lambda_{j,\tau}\right],\]
\[\textsf{b}_j \eqdef \frac{-\mu_j + \sqrt{\mu_j^2+2\bar\lambda_j(1-\alpha)\frac{n_0}{T}}}{2(1-\alpha)\bar\lambda_j},\;\;\; \textsf{B}_j \eqdef \frac{\mu_j + \sqrt{\mu_j^2 +2\underline{\lambda}_j(1-\alpha)}}{2(1-\alpha)\underline{\lambda}_j}.\]

\begin{lemma}\label{lem1}
Fix $j\in\{1,2\}$. For all $\tau\in\mathcal{T}$, $\hat\theta_{j,\tau}\in \M_p^+(\textsf{b}_j,+\infty)$. Let $\{(\theta_1^{(k)},\theta_2^{(k)}),\;k\geq 0\}$ be the output of Algorithm  \ref{algo:1}.  If the step-size $\gamma$ satisfies $\gamma\in(0,\textsf{b}_j^2 ]$, and $\theta_j^{(0)}\in\M_p^+(\textsf{b}_j,\textsf{B}_j)$, then $\theta_j^{(k)}\in\M_p^+(\textsf{b}_j,\textsf{B}_j)$, for all $k\geq 0$.
\end{lemma}
\begin{proof}See Section \ref{sec:proof:lem1}.
\end{proof}

\medskip
\begin{remark}
This lemma is based on Lemma 1, and 2 of \cite{atchade:etal:15} which studied the proximal gradient algorithm for the \textsf{glasso} problem. The first statement of Lemma \ref{lem1} implies that the change-point problem (\ref{est:2}) has at least one solution. The second part shows that when the step-size $\gamma$ is small enough, all the iterates of the algorithm remains positive definite. We note that the fact that $\alpha<1$ is crucial in the arguments. The result remains true where $\alpha=1$, however the arguments is slightly more involved (see \cite{atchade:etal:15}~Lemma 2). For simplicity we focus in this paper on the case $\alpha\in [0,1)$. 
\end{remark}

We now address the issue of convergence. Clearly the function $t\mapsto \mathcal{H}(t\vert \theta_1,\theta_2)$ is not smooth, nor convex. This implies that Algorithm \ref{algo:1} cannot be analyzed using standard optimization tools. And indeed, we will not be able to establish that the output of Algorithm \ref{algo:1} converges to the minimizer $\hat\tau$. Rather, we introduce a containment assumption (Assumption H\ref{H1}) and we show that when it holds, then the output of Algorithm \ref{algo:1} converges to some neighborhood of the true change-point (the existence of this true change-point is part of the assumption).
\begin{assumption}\label{H1}
There exist $\epsilon>0$, $c\geq 0$, $\kappa\in [0,1)$, and $\tau_\star\in\mathcal{T}$ such that the following holds. For any $\tau\in\mathcal{T}$, and for any $\theta_1,\theta_2\in\M_p^+$ such that $\normfro{\theta_1-\hat\theta_{1,\tau}} + \normfro{\theta_2-\hat\theta_{2,\tau}}\leq \epsilon$ we have
\begin{equation}\label{eq:H1}
\left|\textsf{Argmin }_{t\in\mathcal{T}} \H(t\vert \theta_{1},\theta_{2}) -\tau_\star\right| \leq \kappa |\tau - \tau_\star| + c.\end{equation}
\end{assumption}

\begin{remark}\label{rem:thm1}
Plainly, what is imposed in H\ref{H1} is the existence of a time point $\tau_\star\in\mathcal{T}$ (that we can view as the true change-point), such that anytime we take $\tau\in\mathcal{T}$ that is far from $\tau_\star$ in the sense that $|\tau-\tau_\star|>c/(1-\kappa)$, if $\theta_1,\theta_2$ are sufficiently close to the solutions $\hat\theta_{1,\tau}$ and $\hat\theta_{2,\tau}$ respectively, then computing $\textsf{Argmin }_{t\in\mathcal{T}} \H(t\vert \theta_{1},\theta_{2})$ brings us closer to $\tau_\star$:
\[\left|\textsf{Argmin }_{t\in\mathcal{T}} \H(t\vert \theta_{1},\theta_{2}) -\tau_\star\right| \leq \kappa|\tau-\tau_\star| + c < |\tau - \tau_\star|.\] 
This containment assumption is akin to a curvature assumption on the function $t\mapsto \H(t\vert \theta_{1},\theta_{2})$ when $\theta_1$ and $\theta_2$ are reasonably close to $\hat\theta_{1,\tau}$, $\hat\theta_{2,\tau}$, respectively. The assumption seems realistic in settings where the data $X^{(1:T)}$ is indeed drawn from a Gaussian graphical model with true change-point $\tau_\star$, and parameters $\theta_{\star,1}$, $\theta_{\star,2}$. Indeed in this case, and if $T$ is large enough, for any $\tau$ that is not too close to the boundaries, one expect $\hat\theta_{1,\tau}$ and $\hat\theta_{2,\tau}$ to be good estimates of $\theta_{\star,1}$ and $\theta_{\star,2}$, respectively. Therefore if $\normfro{\theta_1-\hat\theta_{1,\tau}} + \normfro{\theta_2-\hat\theta_{2,\tau}}\leq \epsilon$ for $\epsilon$ small enough, one expect as well $\theta_1$ and $\theta_2$ to be close to $\theta_{\star,1}$ and $\theta_{\star,2}$ respectively. Hence $\textsf{Argmin }_{t\in\mathcal{T}} \H(t\vert \theta_{1},\theta_{2})$ should be close to $\textsf{Argmin }_{t\in\mathcal{T}} \H(t\vert \theta_{\star,1},\theta_{\star,2})$, which in turn should be close to $\tau_\star$. Theorem \ref{thm2} below will make this intuition precise.
\begin{flushright}
$\square$
\end{flushright}
\end{remark}
\medskip
In the next result we will see that in fact the iterates $\theta_1^{(k)}$ and $\theta_2^{(k)}$ closely track $\theta_{1,\tau^{(k)}}$ and $\theta_{2,\tau^{(k)}}$ respectively. Hence, when H\ref{H1} holds Equation (\ref{eq:H1}) guarantees that the sequence $\tau^{(k)}$ remains close that $\tau_\star$.
\medskip

\begin{theorem}\label{thm1}
Suppose that $\gamma\in(0,\textsf{b}_1^2\wedge \textsf{b}_2^2]$, and $\theta_j^{(0)}\in\M_p^+(\textsf{b}_j,\textsf{B}_j)$, for $j=1,2$. Then
\[\lim_k\normfro{\theta_1^{(k)}-\hat\theta_{1,\tau^{(k)}}}=0, \;\;\; \lim_k\normfro{\theta_2^{(k)}-\hat\theta_{2,\tau^{(k)}}}=0.\]
Furthermore, if  H\ref{H1} holds then 
\[\limsup_{k\to\infty}\left| \tau^{(k)} - \tau_\star\right| \leq \frac{c}{1-\kappa}.\]
\end{theorem}
\begin{proof}
See Section \ref{sec:proof:thm1}
\end{proof}

We now address the question whether H\ref{H1} is a realistic assumption. More precisely we will show that the argument highlighted in Remark \ref{rem:thm1} holds true under some regularity conditions. Suppose that $X^{(1:T)}\eqdef (X^{(1)},\ldots,X^{(T)})$ are $p$-dimensional independent random variables such that  $X^{(1)},\ldots,X^{(\tau_\star)}\stackrel{i.i.d.}{\sim}\textbf{N}(0,\theta_{\star,1}^{-1})$ and $X^{(\tau_\star + 1)},\ldots,X^{(T)}\stackrel{i.i.d.}{\sim}\textbf{N}(0,\theta_{\star,2}^{-1})$, for some unknown change-point $\tau_\star$,  and unknown symmetric positive definite precision matrices $\theta_{\star,1}\neq \theta_{\star,2}$. We set $\Sigma_{\star,j}\eqdef\theta_{\star,j}^{-1}$,  and we let $s_j$ denotes the number of non-zero entries of $\theta_{\star,j}$, $j=1,2$. For an integer $\iota\in \{1,\ldots,p\}$, we define the $\iota$-th  restricted eigenvalues of $\Sigma_{\star,j}$ as
\begin{multline*}
\underline{\kappa}_j(\iota) \eqdef\inf\left\{u'(\Sigma_{\star,j})u,\;\|u\|_2=1,\;\|u\|_0\leq \iota\right\},\;\\
\;\bar{\kappa}_j(\iota) \eqdef\sup\left\{u'(\Sigma_{\star,j})u,\;\|u\|_2=1,\;\|u\|_0\leq \iota\right\}.\end{multline*}

We set $s\eqdef \max(s_1, s_2)$, $\bar \kappa\eqdef \max\left(\bar\kappa_1(2),\bar\kappa_2(2)\right)$, $\underline{\kappa}\eqdef \min\left(\underline{\kappa}_1(2),\underline{\kappa}_2(2)\right)$, and we set the regularization parameter $\lambda_{j,\tau}$ as 
\begin{equation}\label{eq:lambda}
\lambda_{1,\tau} \eqdef \frac{\bar\kappa}{\alpha T}\sqrt{48\tau\log(pT)},\;
\; \lambda_{2,\tau} \eqdef \frac{\bar\kappa}{\alpha T}\sqrt{48(T-\tau)\log(pT)},\;\;\tau\in\mathcal{T}.\end{equation}
We need to assume that  the parameter $\alpha\in [0,1)$ in the regularization term is large enough to produce approximately sparse solutions in (\ref{glasso_1}). To that end, we assume that
\begin{equation}\label{eq:cond:alpha}
\frac{\alpha}{1-\alpha} \geq \max\left(\|\theta_{\star,1}\|_\infty, \|\theta_{\star,2}\|_\infty\right).\end{equation}
Finally, we assume that the search domain $\mathcal{T}$ is such that for all $\tau\in\mathcal{T}$,
\begin{equation}\label{eq:cond:T:1}
\min\left(\tau,T-\tau\right)\geq A_1^2\log(pT),\end{equation}
where
\[A_1 \eqdef\max\left(2\left(\frac{\bar \kappa}{\underline{\kappa}}\right)^2,(1280)s^{1/2}\bar\kappa(\|\theta_{\star,1}\|_2\vee \|\theta_{\star,1}\|_2)\right),\]
and
\begin{multline}\label{eq:cond:T:2}
\bar\kappa\sqrt{\tau\log(pT)}\geq \frac{1}{2\sqrt{3}}(\tau-\tau_\star)_+ \|\theta_{\star,2}^{-1} - \theta_{\star,1}^{-1}\|_\infty,\;\; \\
\mbox{ and }\;\;\; \bar\kappa\sqrt{(T-\tau)\log(pT)}\geq \frac{1}{2\sqrt{3}}(\tau_\star-\tau)_+ \|\theta_{\star,2}^{-1} - \theta_{\star,1}^{-1}\|_\infty,
\end{multline}
where  $x_+\eqdef \max(x,0)$.
\begin{remark}
Assumption (\ref{eq:cond:T:1}) is a minimum sample size requirement. See for instance \cite{ravikumaretal11}~Theorem 1, and 2 for similar conditions in standard Gaussian graphical model estimation. Here we require to have $\mathcal{T}$ such that $\min(\tau,T-\tau)=O(s\log(pT))$ for all $\tau\in\mathcal{T}$. This obviously implies that we need $T$ to be at least $O(s\log(p))$. It is unclear whether the large constant $1280$ in (\ref{eq:cond:T:1}) is tight or simply an artifact of our proof techniques.

To understand Assumption (\ref{eq:cond:T:2}), note that for $\tau>\tau_\star$, the estimator $\hat\theta_{1,\tau}$ in (\ref{glasso_1}) is based on misspecified data $X^{(\tau_\star+1)},\ldots,X^{(\tau)}$. Hence if $\tau>\tau_\star$ is too far away from $\tau_\star$, the estimators $\hat\theta_{1,\tau}$ may behave poorly, particularly if $\theta_{\star,1}$ are $\theta_{\star,2}$ are very different. Assumption (\ref{eq:cond:T:2}) rules out such settings, by requiring the search domains $\mathcal{T}$ to be roughly a $\sqrt{T}$ neighborhood of $\tau_\star$. Indeed, suppose that $\tau_\star = \rho_\star T$, for some $\rho_\star\in (0,1)$. Then it can be easily checked that any search domain of the form $(\tau_\star-r_1T^{1/2}, \tau_\star + r_2 T^{1/2})$, satisfies (\ref{eq:cond:T:1}) and (\ref{eq:cond:T:2}) for $T$ large enough, provided that
\[0<r_1\leq \frac{2\sqrt{3}\bar \kappa\sqrt{\rho_\star\log(pT)}}{\|\theta_{\star,2}^{-1}-\theta_{\star,1}^{-1}\|_\infty},\;\;\mbox{ and }\;\; 0<r_2\leq \frac{2\sqrt{3}\bar \kappa\sqrt{(1-\rho_\star)\log(pT)}}{\|\theta_{\star,2}^{-1}-\theta_{\star,1}^{-1}\|_\infty}.\]
Of course, this search domain is difficult to use in practice since it depends on $\tau_\star$. In practice, we have found that taking $\mathcal{T}$ of the form $(r T, (1-r)T)$ for $r\leq 0.1$ works well, even though it is much wider than what is prescribed by our theory.
\begin{flushright}
$\square$
\end{flushright}
\end{remark}

For $\tau\in\mathcal{T}$, let
\[r_{1,\tau} \eqdef A_2 \bar\kappa \|\theta_{\star,1}\|_2^2\sqrt{\frac{s_1\log(pT)}{\tau}},\;\;\;r_{2,\tau} \eqdef A_2 \bar\kappa \|\theta_{\star,2}\|_2^2\sqrt{\frac{s_2\log(pT)}{T-\tau}},\]
where $A_2$ is an absolute constant that can be taken as $16\times 20\times \sqrt{48}$. We set $b\eqdef \min(\lambda_{\textsf{min}}(\theta_{\star,1}),\lambda_{\textsf{min}}(\theta_{\star,2}))$, and $B\eqdef\max(\lambda_{\textsf{max}}(\theta_{\star,1}),\lambda_{\textsf{max}}(\theta_{\star,2}))$. We assume that for $j=1,2$, and for $\tau\in\mathcal{T}$,
\begin{multline}\label{eq:tech:cond:thm2}
r_{j,\tau}\leq\min\left(\frac{ \lambda_{\textsf{min}}(\theta_{\star,j})}{4},\frac{\|\theta_{\star,j}\|_\infty}{2},\frac{\|\theta_{\star,j}\|_1}{1+8s_j^{1/2}}\right),\;\;\;\; r_{j,\tau}\leq \frac{\|\theta_{\star,2}-\theta_{\star,1}\|_{\textsf{F}}}{2(1+8s^{1/2})}\\
\mbox{ and }\;\;\; r_{j,\tau} \leq A_2\left(\frac{b}{B}\right)^4\frac{\|\theta_{\star,j}\|_1}{s_j^{1/2}}.
\end{multline}
\begin{remark}
Condition (\ref{eq:tech:cond:thm2}) is mostly technical. As we will see below in Lemma \ref{lem2:thm2}, the term $r_{j,\tau}$ is the convergence rate toward $\theta_{\star,j}$ of the estimator $\hat\theta_{j,\tau}$. Note that all the terms on the right-hand sides in (\ref{eq:tech:cond:thm2}) depend only on $\theta_{\star,1}$ and $\theta_{\star,2}$. Hence if $s_j$ and the norms of $\theta_{\star,1}$, $\theta_{\star,2}$, $\theta_{\star,2}-\theta_{\star,1}$  do not grow with $p$, and $r_{j,\tau}\to 0$ as $p,T\to\infty$, then it is clear that (\ref{eq:tech:cond:thm2}) holds for $T$ large enough.
\end{remark}

\begin{theorem}\label{thm2}
Consider the output $\{(\theta_1^{(k)},\theta_2^{(k)}),\;k\geq 0\}$ of Algorithm  \ref{algo:1}. Suppose that $\gamma\in(0,\textsf{b}^2_1\wedge \textsf{b}_2^2]$, and $\theta_j^{(0)}\in\M_p^+(\textsf{b}_j,\textsf{B}_j)$, for $j=1,2$. Suppose that the statistical model underlying the data $X^{(1:T)}$ is as above, and that (\ref{eq:lambda})-(\ref{eq:tech:cond:thm2}) hold. Suppose also that 
\begin{equation}\label{eq:idenf}
\|\theta_{\star,2} -\theta_{\star,1}\|_{\textsf{F}} \geq 8A_2\max\left[\left(\frac{\lambda_{\textsf{min}}(\theta_{\star,1})}{\lambda_{\textsf{max}}(\theta_{\star,1})}\right)^2\frac{\|\theta_{\star,1}\|_1}{s_1^{1/2}}, \left(\frac{\lambda_{\textsf{min}}(\theta_{\star,2})}{\lambda_{\textsf{max}}(\theta_{\star,2})}\right)^2\frac{\|\theta_{\star,2}\|_1}{s_2^{1/2}}\right].\end{equation}
Then
\begin{equation}\label{eq:limsup}
\limsup_{k\to\infty} \left|\tau^{(k)} -\tau_\star\right| \leq \frac{4}{C_0} \log(p),\end{equation}
with probability at least $1-\frac{8}{pT} - \frac{4}{p^2\left(1-e^{-C_0}\right)}$,
where
\[C_0 \eqdef   \min\left[\frac{\|\theta_{\star,2}-\theta_{\star,1}\|_{\textsf{F}}^4}{128 B^4\|\theta_{\star,2}-\theta_{\star,1}\|_1^2},\left(\frac{\underline{\kappa}}{\bar\kappa}\right)^4\right].\]
\end{theorem}
\begin{proof}
See Section \ref{sec:proof:thm2}.
\end{proof}

\begin{remark}
The main point of the theorem is that under the assumptions and data generation mechanism described above, the containment assumption H\ref{H1} holds with probability as least $1-\frac{8}{pT} - \frac{4}{p^2\left(1-e^{-C_0}\right)}$, and where $\epsilon$ can be taken as $\min_\tau r_{1,\tau}\wedge r_{2,\tau}/\sqrt{p}$, $\kappa=0$, and $c = 4\log(p)/C_0$. Conclusion (\ref{eq:limsup}) is then simply a consequence of  Theorem \ref{thm1}.
\end{remark}

\subsection{A stochastic version}
When $T$ is much larger than $p$, Step 3 of Algorithm \ref{algo:1} becomes costly. In such cases, one can gain in efficiency by replacing Step 3 by a Monte Carlo approximation. We explore the use of simulated annealing to approximately solve Step 3 of Algorithm \ref{algo:1}. Given $\theta_1,\theta_2\in\M_p$, and $\beta>0$, let $\pi_{\beta,\theta_1,\theta_2}$ denote the probability distribution on $\mathcal{T}$ defined as
\[\pi_{\beta,\theta_1,\theta_2}(\tau) =\frac{1}{Z_{\beta,\theta_1,\theta_2}} \exp\left(-\frac{\mathcal{H}(\tau\vert\theta_1,\theta_2)}{\beta}\right),\;\;\tau\in \mathcal{T}.\]
Here, $Z_{\beta,\theta_1,\theta_2}$ is the normalizing constant, and $\beta>0$ is the cooling parameter, that we shall drive down to zero with the iteration to increase the accuracy of the Monte Carlo approximation. Direct sampling from $\pi_{\beta,\theta_1,\theta_2}$ is typically possible, but this has the same computational cost as Step 3 of Algorithm \ref{algo:1}. We will use a Markov Chain Monte Carlo approach which will allow us to make only a small number of calls of the function $\mathcal{H}$, per iteration. Let $\mathcal{K}_{\beta,\theta_1,\theta_2}$ denote a Markov kernel on $\mathcal{T}$ with invariant distribution $\pi_{\beta,\theta_1,\theta_2}$. Typically we will choose $\mathcal{K}_{\beta,\theta_1,\theta_2}$ as a Metropolis-Hastings Markov kernel (we give examples below).

We consider the following algorithm. As in Algorithm \ref{algo:1}, $\gamma$ is a given step-size.  We choose a  decrease sequence of temperature $\beta^{(k)}$ that we use along the iterations.

\begin{algorithm}\label{algo:2}
Fix a step-size $\gamma>0$, and a cooling sequence $\{\beta^{(k)}\}$. Pick some initial value $\tau^{(0)}\in\mathcal{T}$, $\theta_1^{(0)},\theta_2^{(0)}\in\M_p^+$.  Repeat for $k=1,\ldots,K$. Given ($\tau^{(k-1)}$, $\theta_1^{(k-1)}$, $\theta^{(k-1)}_2$), do the following:
\begin{enumerate}
\item Compute
\[\theta_1^{(k)} = \Prox_{\gamma\lambda_{1,\tau^{(k-1)}}}\left(\theta_1^{(k-1)}-\gamma\left(S_1(\tau^{(k-1)})-(\theta_1^{(k-1)})^{-1}\right)\right),\]
\item compute
\[\theta_2^{(k)} = \Prox_{\gamma\lambda_{2,\tau^{(k-1)}}}\left(\theta_2^{(k-1)}-\gamma\left(S_2(\tau^{(k-1)})-(\theta_2^{(k-1)})^{-1}\right)\right),\]
\item draw
\[\tau^{(k)}\sim \mathcal{K}_{\beta^{(k)},\theta_1^{(k)},\theta_2^{(k)}}(\tau^{(k-1)},\cdot).\]
\end{enumerate}
\begin{flushright}
$\square$
\end{flushright}
 \end{algorithm}

For most commonly used MCMC kernels, each iteration of Algorithm \ref{algo:2} has a computatinal cost of $O(p^3)$, which is better than $O(p^3 +Tp^2)$ needed by Algorithm \ref{algo:1}, when $T\geq p$. However Algorithm \ref{algo:2} travels along the change-point space $\mathcal{T}$ more slowly. Hence overall, a larger number of iterations would typically be needed for Algorithm \ref{algo:2} to converge. Even after accounting for this slow convergence, Algorithm \ref{algo:2} is still substantially faster than Algorithm \ref{algo:1}, as shown in Table  \ref{table:runtime_100} and \ref{table:runtime_500}. A rigorous analysis of the convergence of Algorithm \ref{algo:2} is beyond the scope of this work, and it left as a possible future research.




%
%
%

\subsection{Extension to multiple change-points}\label{sec:bs}
We extend the method to multiple change-points by binary segmentation. Binary segmentation is a standard method for detecting multiple change-points. The method proceeds by first searching for a single change-point. When a change-point is found the data is split into the two parts defined by the detected change-point. A similar search is then performed on each segment which can result in further splits. This recursive procedure continues until a certain stopping criterion is satisfied. Here we stop the recursion if

\[\ell_{\tau} +Cp \geq  \ell_F ,\]

where $\ell_{\tau}$ is the penalized negative log-likelihood obtained with the additional change-point $\tau$, and $\ell_F$ is the penalized negative log-likelihood without the change-point. The term $Cp$ is a penalty term for model complexity, where $C$ is a user-defined parameter. As we show in the simulations, values of $C$ between $(0,4)$ seem to produce the best results in our setting.  

\section{Numerical experiments}\label{sec:num:exp}
We investigate the different algorithms presented here in a variety of settings.  For all the algorithms investigated the choice of the step-size $\gamma$ and the regularizing parameter  $\lambda$ are important.  For all experiments, and as suggested by (\ref{eq:lambda}), we found that setting $\lambda_{1,\tau} = \lambda \sqrt{\frac{\log\{p\}}{\tau}}$ and $\lambda_{2,\tau} = \lambda \sqrt{\frac{\log\{p\}}{T - \tau}}$ worked well.  For the time-comparison in Section \ref{sec:comp}  we used $\lambda = 0.1$ and $\gamma = 3.5$  when $T = 1000$, and we used $\lambda = 0.01$ and $\gamma = 3.5$ when $T = 500$.  For the remainder of the experiments we set $\lambda = 0.13$ and $\gamma = 0.25$.  For the minimum sample size $n_0$, we found that taking $n_0$ from  $\{0.01T, 0.05T, 0.1T\}$ worked well.

We initialize $\tau^{(0)}$ to a randomly selected value in $\mathcal{T}$. The initial value $\theta_1^{(0)}$ and $\theta_2^{(0)}$ are taken as $\theta_j^{(0)} = (S_j(\tau^{(0)}) + \epsilon I)^{-1}$ where $\epsilon$ is a constant chosen to maintain positive definiteness. For cases where $p < \tau$ and $p < T - \tau$ we used $\epsilon = 0$, while for larger values of $p$ we set $\epsilon = 0.2$.

For the data generation in the simulations,  we typically choose $\tau_\star=T/2$ unless otherwise specified, and unless otherwise specified, we generate independently the matrices $\theta_{\star,1}$ and $\theta_{\star,2}$ as follows. First we generate a random symmetric sparse matrix $M$ such that the proportion of non-zero entries is 0.25.  We add 4 to all positive entries and subtract 4 from all negative entries.  Then we set the actual precision matrix as $\theta_{\star,j} = M + (1 - \lambda_{\min}(M))I_p$ where $\lambda_{\min}(M)$ is the smallest eigenvalue of $M$. The resulting precision matrices contain roughly $25\%$ non-zero off-diagonal elements. For each simulation a new pair of precision matrices was generated as well as the corresponding data set.

For Algorithm \ref{algo:2} we also experimented with a number of MCMC kernel $\mathcal{K}_{\beta,\theta_1,\theta_2}$. We experiment with the independence Metropolis sampler with proposal $\textbf{U}(n_0, T - n_0)$. We  also tried a Random Walk Metropolis with a truncated Gaussian proposal $\textbf{N}(\tau^{(k-1)}, \sigma^2)$, for some scale parameter $\sigma>0$. Finally, we also experimented with a mixture of these two Metropolis-Hastings kernels. We found that for our simulations the Independent Metropolis kernel works best, although the mixture kernel also performed well. For the cooling schedule of simulated annealing we use  $\beta^{(0)} = 1$, and a geometric decay $\beta^{(n)} = \alpha \beta^{(n-1)}$ with $\alpha = \left(\frac{\beta^{(M)}}{\beta^{(0)}}\right)^{1/M}$ where $\beta^{(M)}=0.001$, and $M$ is the maximum number of iterations.

\subsection{Time comparison}\label{sec:comp}

First we compare the running times of the proposed algorithms and the brute force approach. We consider two settings: $(p=100,T=1000)$ and $(p=500,T=500)$. In the setting $(p=100,T=1000)$, 100 independent runs of Algorithms \ref{algo:1} and \ref{algo:2} are performed and the average run-times are reported in Table \ref{table:runtime_100}. In the setting $(p=500,T=500)$ 10 independent runs of Algorithms \ref{algo:1} and \ref{algo:2} are used, and the results are presented in Table \ref{table:runtime_500}.   We compare these times to results from one simulation run of the brute-force approach. 

We consider two  stopping criterion for Algorithm \ref{algo:1} or \ref{algo:2}. The first criteria stops the iterations of
\[\frac{1}{T}|\tau^{(k)} - \tau_\star| < 0.005 \;\;\mbox{ and }\;\; \frac{\|\theta_1^{(k)} - \hat{\theta}_1\|_F}{\|\hat{\theta}_1\|_F} + \frac{\|\theta_2^{(k)} - \hat{\theta}_2\|_F}{\|\hat{\theta}_2\|_F} < 0.05,\;\;\;\tag{V1}\]  
where $\hat{\theta}_1$ and $\hat{\theta}_2$ are obtained by performing 1000 proximal-gradient steps at the true $\tau$ value. An interesting feature of the proposed approximate MM algorithms is that the change-point sequence $\tau^{(k)}$ can converge well before $\theta^{(k)}_1$ and $\theta_2^{(k)}$. To illustrate this, we also explore the alternative approach of stopping the iterations only based on $\tau^{(k)}$,  namely when
\[\frac{1}{T}|\tau^{(k)} - \tau_\star| < 0.005.\;\;\;\tag{V2}\]
Finally, we note that we implement the brute force approach by running $500$ proximal-gradient steps for each possible value of $\tau$. Note that $500$ iterations is typically smaller than the number of iterations needed to satisfy (V1). 

Tables \ref{table:runtime_100} and \ref{table:runtime_500} highlight the benefits of Algorithm \ref{algo:1} and Algorithm \ref{algo:2} as the run-time is several orders of magnitude lower than the brute force approach.  Additionally, while Algorithm \ref{algo:2} requires more iterations than Algorithm \ref{algo:1} its run-time is typically smaller.  The benefits of Algorithm \ref{algo:2} are particularly clear for large values of $p$ and $T$ (under stopping criterion (V1)). The stopping criteria (V2) highlights the fact that the $\tau^{(k)}$ sequence in the proposed algorithms can converge well before the $\theta$-sequences. 


\begin{table}
\begin{tabular}{l l | c c c }
Variant & & Brute Force & Approx. MM & Simulated Annealing \\
\hline
(V1) & Time (Seconds) & 550.34 & 160.05 & 7.02 \\
 & Iterations & - & 573.41 & 598.71 \\
\hline
(V2) & Time (Seconds) & - & 3.16 &  2.96 \\
 & Iterations & - & 1.12 & 100.20 \\
\hline
\end{tabular}
\caption{Run-time Comparison $(p=100,T=1000)$}
\label{table:runtime_100}
\end{table}

\begin{table}
\begin{tabular}{l l | c c c }
Variant & & Brute Force & Approx. MM & Simulated Annealing \\
\hline
(V1) & Time (Seconds) & 19205.61 & 7017.64 & 258.51 \\
 & Iterations & - & 961.40 & 962.20 \\
\hline
(V2) & Time (Seconds) & - & 187.36 &  167.21 \\
 & Iterations & - & 1.90 & 131.10 \\
\hline
\end{tabular}
\caption{Run-time Comparison $(p=500,T=500)$}
\label{table:runtime_500}
\end{table}

\subsection{Behavior of the algorithm when the change-point is at the edge}
We investigate how the brute force algorithm, Algorithm \ref{algo:1}, and Algorithm \ref{algo:2} perform when change-points are non-existent or close to the edges.  The results for the brute force algorithm are presented in Figure \ref{fig:bf_edge_cp}, the results for Algorithm \ref{algo:1} are presented on Figure \ref{fig:ro_edge_cp} and the results for Algorithm \ref{algo:2} are presented on Figure \ref{fig:sa_edge_cp}.  For Algorithm \ref{algo:1} and Algorithm \ref{algo:2} the figure contains two subfigures, the first showing the trajectories of the sequences $\{\tau^{(k)}\}$ produced by the algorithm, and the second showing a histogram of the final location of the estimated $\tau$, based on $200$ replications.  Additionally, a line is included to show the location of the true $\tau$.  For the brute force algorithm the trace plot is removed.  The results suggest that Algorithm \ref{algo:1} and Algorithm \ref{algo:2} have more trouble when the true $\tau$ is close to the edge of the sample.


\begin{figure}[ht]
\centering     
\subfigure[No change-point]{\label{fig:a}\includegraphics[width=0.45\textwidth]{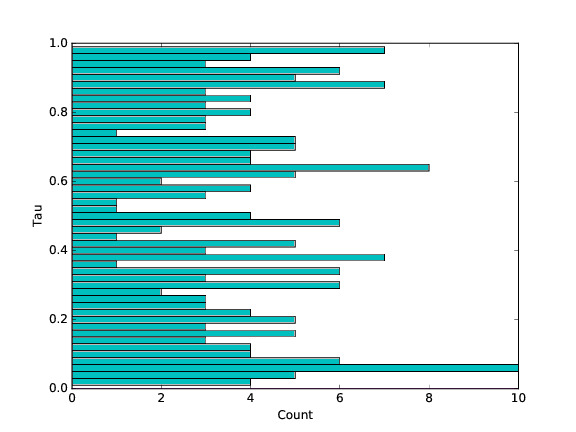}}
\subfigure[Change-point at $\tau = 0.1T$]{\label{fig:b}\includegraphics[width=0.45\textwidth]{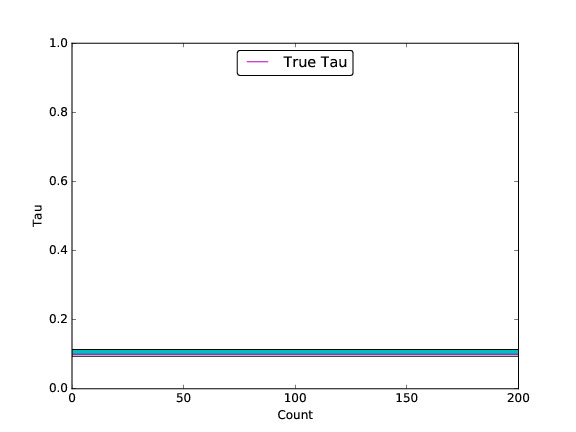}}
\subfigure[Change-point at $\tau = 0.25T$]{\label{fig:c}\includegraphics[width=0.45\textwidth]{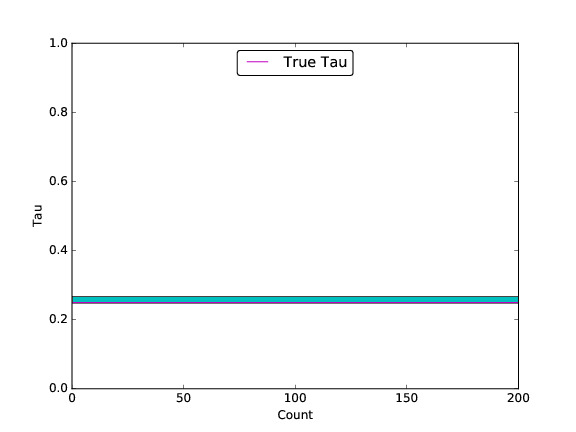}}
\subfigure[Change-point at $\tau = 0.5T$]{\label{fig:d}\includegraphics[width=0.45\textwidth]{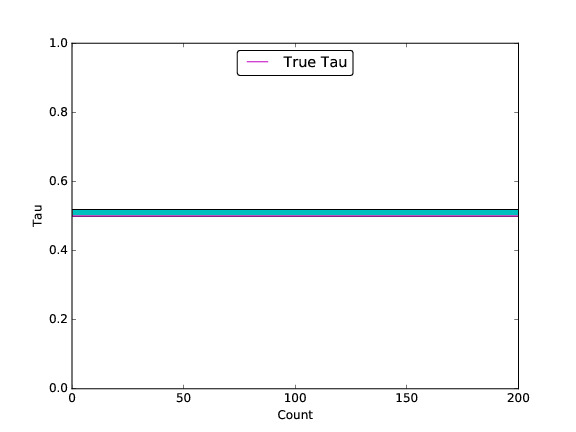}}
\caption{Change-point close to the edge. Results for the brute force approach.}
\label{fig:bf_edge_cp}
\end{figure}

\begin{figure}[ht]
\centering     
\subfigure[No change-point]{\label{fig:a}\includegraphics[width=0.45\textwidth]{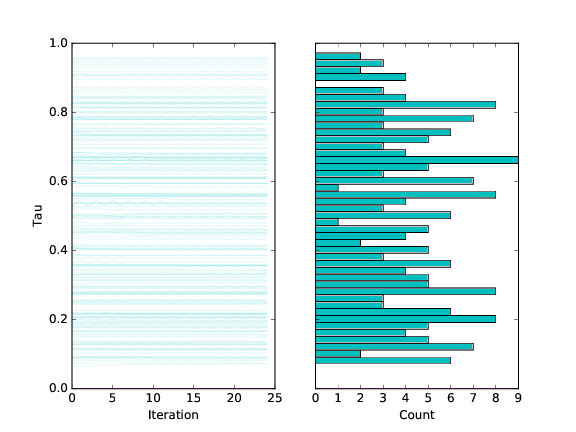}}
\subfigure[Change-point at $\tau = 0.1T$]{\label{fig:b}\includegraphics[width=0.45\textwidth]{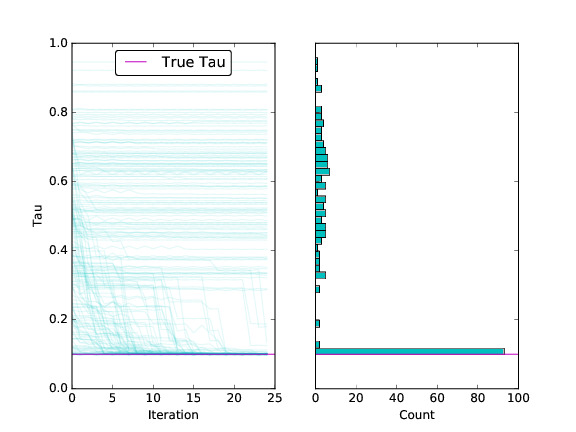}}
\subfigure[Change-point at $\tau = 0.25T$]{\label{fig:c}\includegraphics[width=0.45\textwidth]{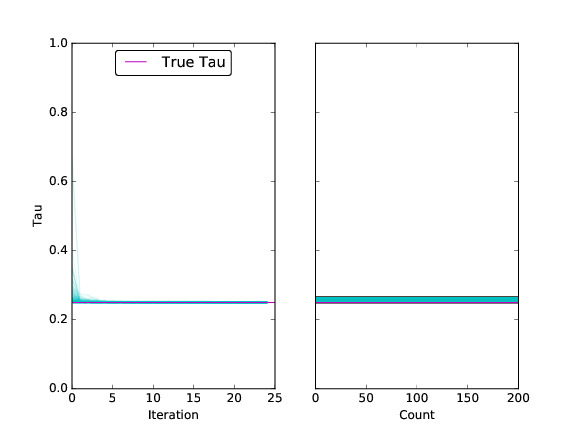}}
\subfigure[Change-point at $\tau = 0.5T$]{\label{fig:d}\includegraphics[width=0.45\textwidth]{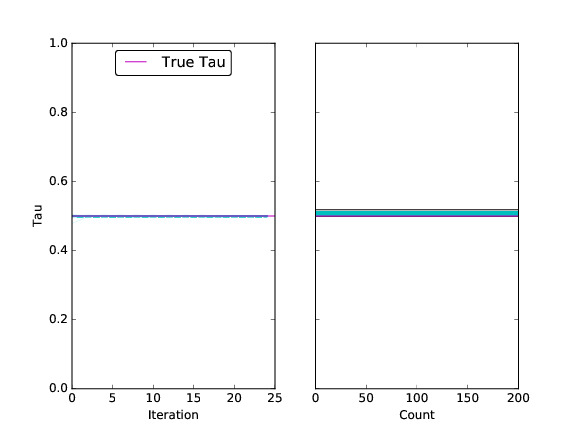}}
\caption{Change-point close to the edge. Results for Algorithm \ref{algo:1}.}
\label{fig:ro_edge_cp}
\end{figure}

\begin{figure}[ht]
\centering     
\subfigure[No change-point]{\label{fig:sa_a}\includegraphics[width=0.45\textwidth]{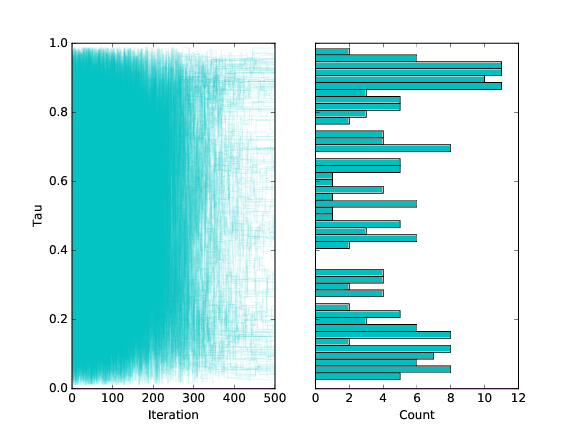}}
\subfigure[Change-point at $\tau = 0.1T$]{\label{fig:b}\includegraphics[width=0.45\textwidth]{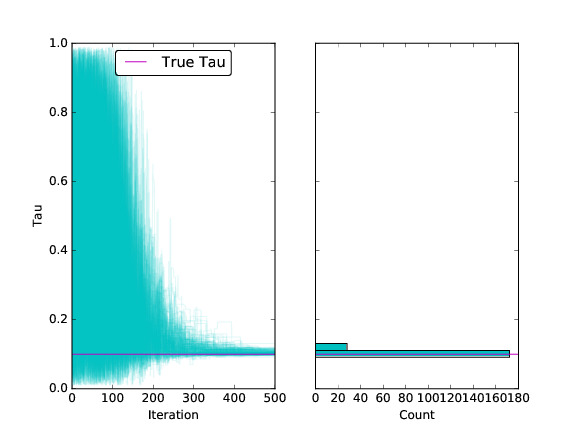}}
\subfigure[Change-point at $\tau = 0.25T$]{\label{fig:c}\includegraphics[width=0.45\textwidth]{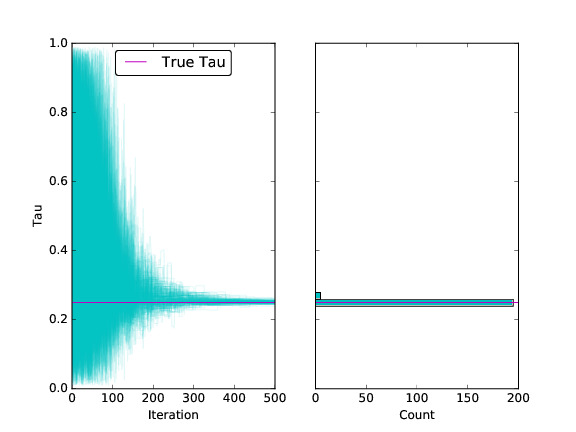}}
\subfigure[Change-point at $\tau = 0.5T$]{\label{fig:d}\includegraphics[width=0.45\textwidth]{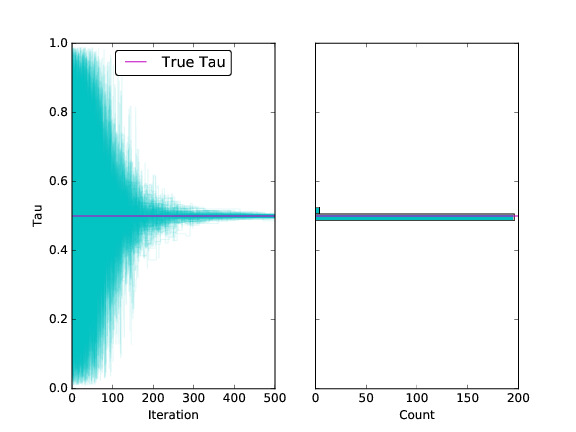}}
\caption{Change-point close to the edge. Results for Algorithm \ref{algo:2}.}
\label{fig:sa_edge_cp}
\end{figure}

\clearpage

\subsection{Behavior of the algorithms when $\theta_1$ and $\theta_2$ are similar}
As $\theta_1$ and $\theta_2$ get increasingly similar, the location of the change-point becomes increasingly more difficult to find. We investigate the behavior of the proposed algorithms in such settings. We generate the true precision matrices $\theta_1$ and $\theta_2$ as follows. We draw a random precision matrix $\theta$ with $q\%$ non-zero off-diagonal elements, and $C_1$ and $C_2$ two random precision matrix with $p\%$ non-zero off-diagonal elements. We choose $C_1$ and $C_2$ to have the same diagonal elements. Then we set  $\theta_1 = \theta + C_1$ and $\theta_2 = \theta + C_2$, which are then used to generate the dataset for the experiment.  
The ratio $p/q$ is a rough indication of the signal. See Figure \ref{fig:bf_sim_comp} for a comparison of the performance for different values for $q$ and $p$ for the brute force algorithm, Figure \ref{fig:ro_sim_comp} for Algorithm \ref{algo:1}, and Figure \ref{fig:sa_sim_comp} for Algorithm \ref{algo:2}.

\begin{figure}[h]
\centering     
\subfigure[$q = 25$, $p = 0$]{\label{fig:a}\includegraphics[width=0.45\textwidth]{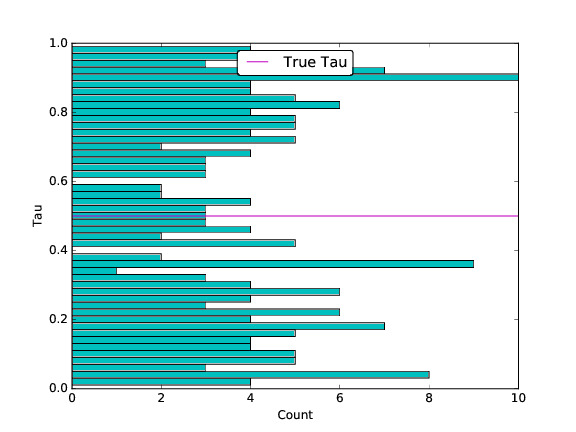}}
\subfigure[$q = 17.5$, $p = 7.5$]{\label{fig:c}\includegraphics[width=0.45\textwidth]{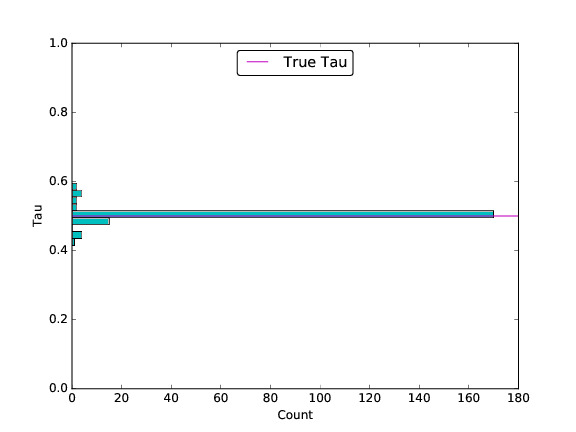}}
\subfigure[$q = 10$, $p = 15$]{\label{fig:d}\includegraphics[width=0.45\textwidth]{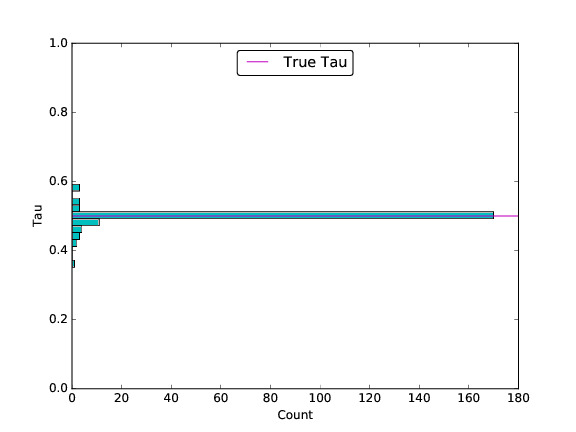}}
\subfigure[$q = 0$, $p = 25$]{\label{fig:e}\includegraphics[width=0.45\textwidth]{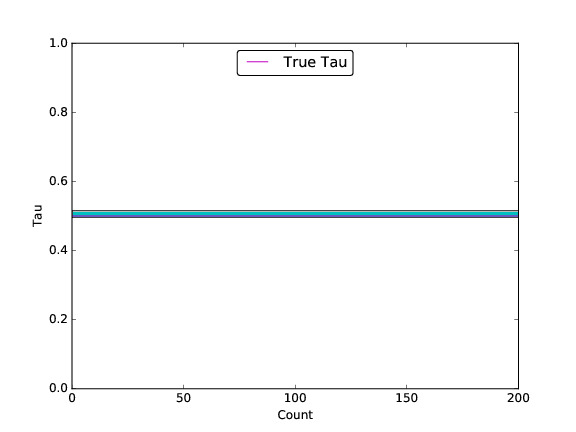}}
\caption{Behavior of the brute force approach when $\theta_1$ and $\theta_2$ are similar.}
\label{fig:bf_sim_comp}
\end{figure}

\begin{figure}[h]
\centering     
\subfigure[$q = 25$, $p = 0$]{\label{fig:a}\includegraphics[width=0.45\textwidth]{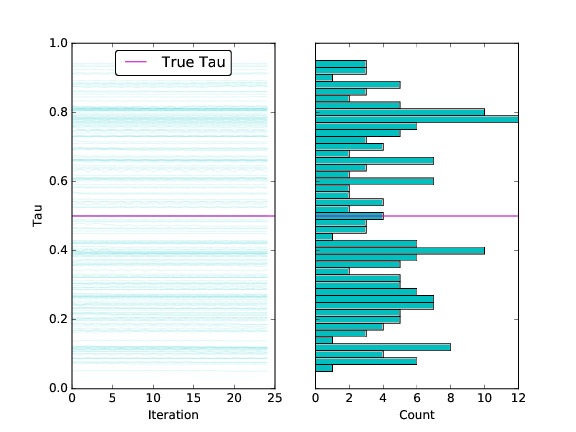}}
\subfigure[$q = 17.5$, $p = 7.5$]{\label{fig:c}\includegraphics[width=0.45\textwidth]{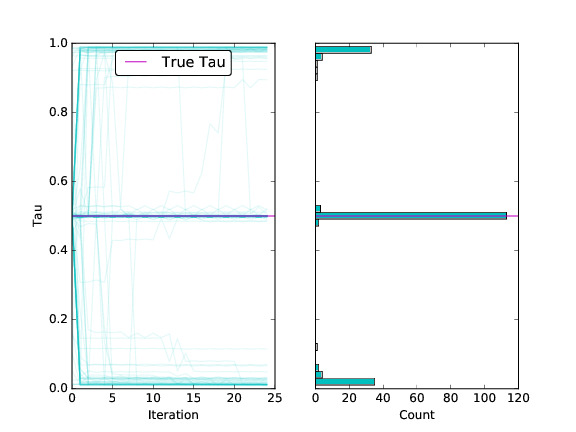}}
\subfigure[$q = 10$, $p = 15$]{\label{fig:d}\includegraphics[width=0.45\textwidth]{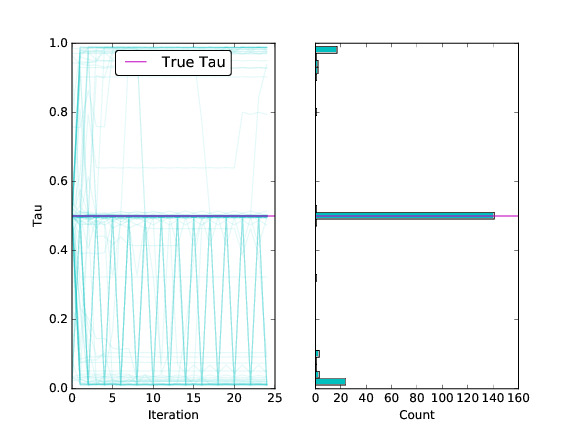}}
\subfigure[$q = 0$, $p = 25$]{\label{fig:e}\includegraphics[width=0.45\textwidth]{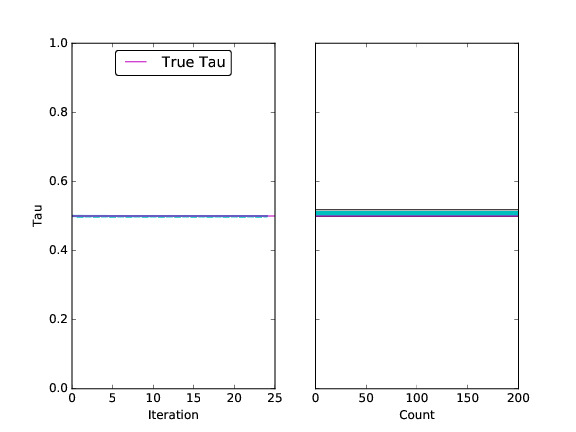}}
\caption{Behavior of Algorithm \ref{algo:1} when $\theta_1$ and $\theta_2$ are similar.e}
\label{fig:ro_sim_comp}
\end{figure}

\begin{figure}[h]
\centering     
\subfigure[$q = 25$, $p = 0$]{\label{fig:a}\includegraphics[width=0.45\textwidth]{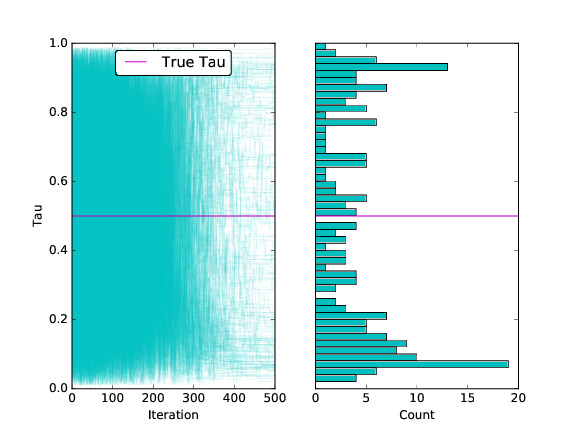}}
\subfigure[$q = 17.5$, $p = 7.5$]{\label{fig:c}\includegraphics[width=0.45\textwidth]{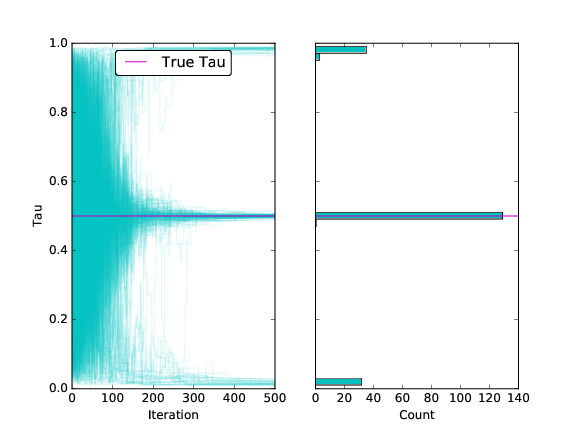}}
\subfigure[$q = 10$, $p = 15$]{\label{fig:d}\includegraphics[width=0.45\textwidth]{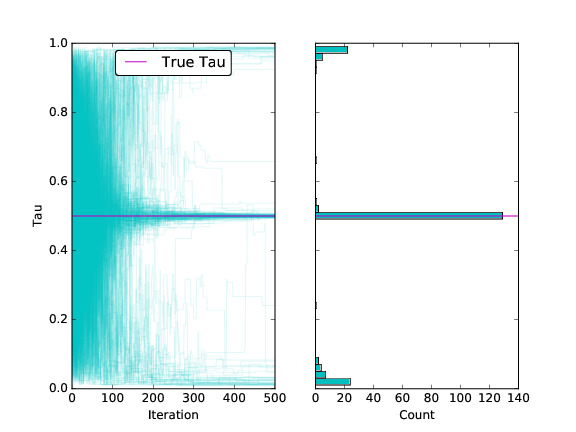}}
\subfigure[$q = 0$, $p = 25$]{\label{fig:e}\includegraphics[width=0.45\textwidth]{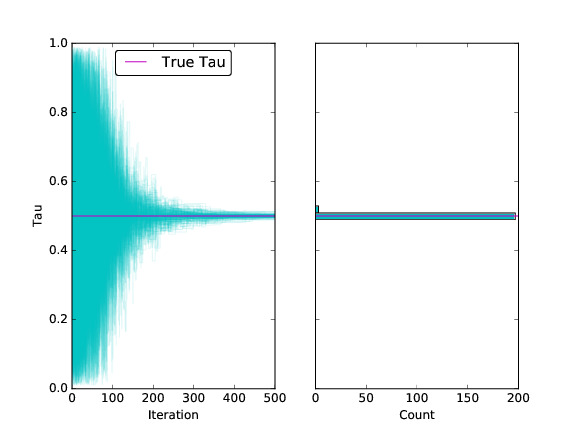}}
\caption{Behavior of Algorithm \ref{algo:2} when $\theta_1$ and $\theta_2$ are similar.}
\label{fig:sa_sim_comp}
\end{figure}

\subsection{Sensitivity to the stopping Criteria in binary segmentation}
This section considers the stopping condition for the binary segmentation algorithm (see Section \ref{sec:bs}) and how it performs with different configurations.  A condition is required for determining when the binary segmentation splitting should reject a change-point and stop running.  The stopping condition that we use is the following, stop if

\[\ell_{\tau} +Cp \geq  \ell_F ,\]

where $\ell_{\tau}$ is the penalized negative log-likelihood obtained with the additional change-point $\tau$, and $\ell_F$ is the penalized negative log-likelihood without the change-point. The term $C$ is a user-defined parameter.

As mentioned above, the proposed algorithms can diverge when the step-size $\gamma$ is not appropriately selected. Tuning $\gamma$ in the binary segmentation setting  presents some challenge since the splitting of the data can result in data segments with very different lengths. Here we have chosen not to tune $\gamma$ to the data segment, and to stop the binary segmentation splitting if the sequence $\hat{\theta}_1^{(k)}$ or $\hat{\theta}_2^{(k)}$ appear to diverge. We  found that stopping the algorithm when $||\hat{\theta_i}^{(k)}||_2^2 > 2\times 10^3$ was sufficient for our data.

In the binary segmentation, since the estimates of $\theta_1$ and $\theta_2$ may not have converged by the end of the search for $\tau$ it may be worth continuing the estimation procedure for $\theta_1$ and $\theta_2$ so that the resulting penalized log-likelihoods are comparable.  Hence after each split from the binary segmentation search, we perform an additional 500 iterations to estimate $\theta_1$ and $\theta_2$ at the resulting $\tau$.  


See Figure \ref{fig:tau_config} for a series of heatmaps showing how often the binary segmentation method finds a given number of change-points for different values of $C$.  These results suggest  that the choice of $C$ in the interval $(0,4)$ is reasonable.  These results are produced using Algorithm \ref{algo:2} for speed, however, the results are identical for the other two algorithms considered.


\begin{figure}[h]
\centering     
\includegraphics[width=\textwidth]{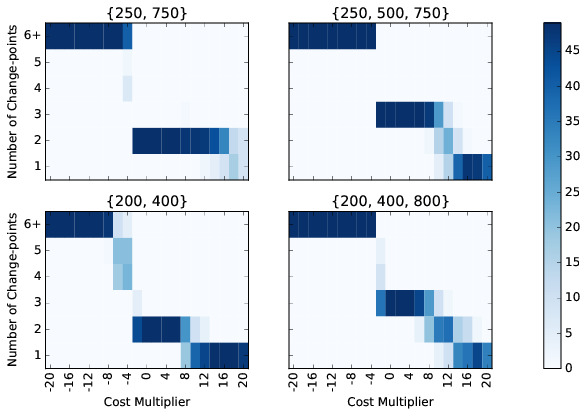}
\caption{Number of change-points detected by binary segmentation as function of the cost multiplier $C$. The number of true change-points is indicated on top of the plots.}
\label{fig:tau_config}
\end{figure}

\subsection{Large scale experiments}
We also investigate the behavior of the proposed algorithms for larger values of $p$. We performed several (100) runs of Algorithm \ref{algo:2} for $T=1000$, and $p\in\{500, 750, 1000\}$. From these $100$ runs we estimate the distributions of the iterates (by boxplots) after $10,100,200,\ldots,1000$ iterations.  The results are presented in figure \ref{fig:sp_larger_p}.  The results show again a very quick convergence toward $\tau_\star$.

\begin{figure}[h]
\centering     
\includegraphics[width=\textwidth]{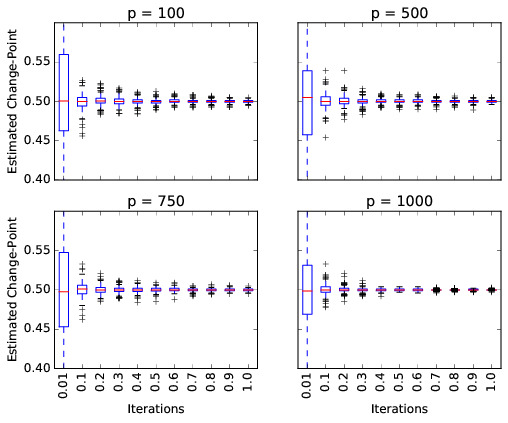}
\caption{Change-point Estimates for Larger $p$}
\label{fig:sp_larger_p}
\end{figure}

\subsection{A real data analysis}\label{sec:real:data}

In finance and econometrics there is considerable interest in regime-switching models in the context of volatility, particularly because these switches may correspond to real events in the economy (\cite{banerjeeurga05,beltrattimorana06,gunay14,choietal10}).  However, much of the literature is limited to the low dimensional case, due to the difficulty involved in estimating change-points for higher dimensions.  We are able to use our method extend this work by estimating change-points in the covariance structure of the S\&P 500.

Data from the S\&P 500 was collected for the period from 2000-01-01 to 2016-03-03.  From this initial sample a subset of tickers was selected for which at least 3000 corresponding observations exist.  This produced a sample extending from 2004-02-06 to 2016-03-03, consisting of 3039 observations and 436 tickers.  We follow the data cleaning procedure from \cite{laffertyetal12}.  For each ticker we generate the log returns $\log{\frac{X_t}{X_{t-1}}}$ and standarizing the resulting returns.   We then threshold any values more than three standard deviations away from the mean.

See Figure \ref{fig:full_cp} for a plot of the binary segmentation search path.  For each segment, the corresponding simulated annealing algorithm was run 50 times to produce a plot of the trace.  The blue line in each plot shows the selected change-point, while the red lines show the edge of the searched segment.  The cyan lines show the trace for each simulated annealing run.  For this setting $\lambda = 0.002$ and $\gamma = 0.5$.  We initialize $\hat{\theta}^{(0)} = (S(\tau^{(0)}) + I \epsilon)^{-1}$ where $\epsilon = 10^{-4}$ and $\tau^{(0)}$ is selected randomly.  After the simulated annealing run the proximal gradient algorithm was run an additional 2000 steps, to produces estimates of $\theta_1$ and $\theta_2$.  Here we increase the step-size to $\gamma = 350$ to accelerate the convergence.  For the binary segmentation we found that selecting the threshold constant, $C = 0.005$, found a reasonable set of change-points.

\begin{figure}[h]
\centering     
\includegraphics[width=\textwidth]{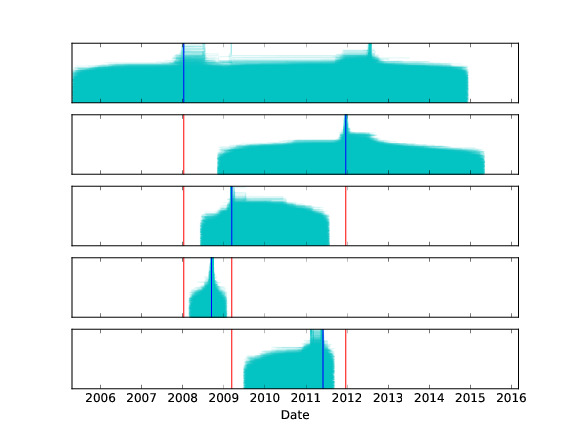}
\caption{Simulated Annealing Trace}
\label{fig:full_cp}
\end{figure}

We next look at how well the estimated change-points correspond to real world events.  Our change-point set seems to do a good job of capturing both the Great Recession and a fall in stock prices during August of 2011 related to the European debt crisis and the downgrading of United State's credit-rating.  The first change-point in our set is January 11th 2008.  The National Bureau of Economic Research (NBER) identifies December of 2007 as the beginning of the Great Recession, which this change-point seems to capture.  Additionally, 10 days after the change-point, the Financial Times Stock Exchange (FTSE) would experience its biggest fall since September 11th 2001.  The second change-point occurred on September 15th 2008, the day on which Lehman Brothers filed for bankruptcy protection, one of the key events of the Great Recession.  The third change-point takes place on March 16th 2009, corresponding to the end of the bear market in the United States.  To get a better sense of the importance of the fourth and fifth change-points see Figure \ref{fig:ted}.  Figure \ref{fig:ted} shows a plot of all the change-points overlaid on the TED spread for our sample.  The TED spread corresponds to the difference between the 3-year LIBOR rate and the 3-year T-bill interest rate.  It is commonly used as a measure of the general credit risk of the economy.  The fourth change-point, on June 1st 2011, and the fifth change-point, on December 21st 2011, likely capture a period of heightened concerns over the possible spread of the European debt crisis to Spain and Italy, during August of 2011.  This period also saw the downgrading of the S\&P's credit rating of the United States from AAA to AA+.  The fourth and fifth change-points, bookend a period of increase in the TED spread, corresponding to these events.

\begin{figure}[h]
\centering     
\includegraphics[width=\textwidth]{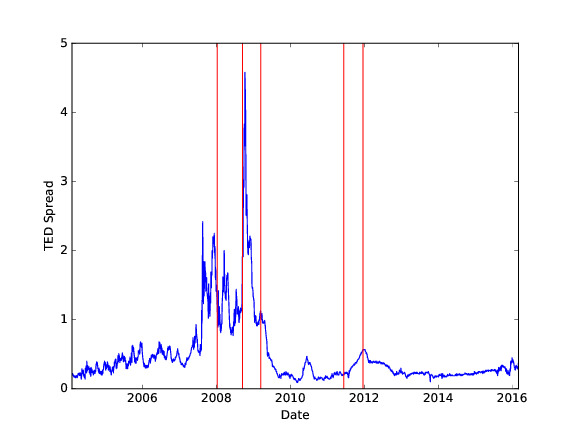}
\caption{TED Spread}
\label{fig:ted}
\end{figure}


Given that the change-point set identified seems sensible, we then investigate what the corresponding $\hat{\theta}$ estimates look like, and whether any interesting conclusions can be drawn from our estimates.  See Figure \ref{fig:segments} for a plot of the adjacency matrix for each $\hat{\theta}$ estimate.  The yellow boxes correspond to Global Industry Classification Standard (GICS) sectors.  These results tell an intuitive story about how the economy behaves during financial crises.  Following both the collapse of Lehamn Brother's and the events of August 2011, we see a dramatic increase in connectivity between returns even outside of GICS sectors.  To get a better sense of this see Figure \ref{fig:indedges} for a similar series of plots where edges are summed over each sector.  Figure \ref{fig:fin_edges} gives an expanded version of the summed edge plot for the first $\hat{\theta}$ estimate, as well as the corresponding sector labels for reference.  Again, we can see that during periods of crisis, the off diagonal elements --corresponding to edges between different sectors -- become more significant than during periods of general stability.


\begin{figure}[h]
\centering     
\includegraphics[width=\textwidth]{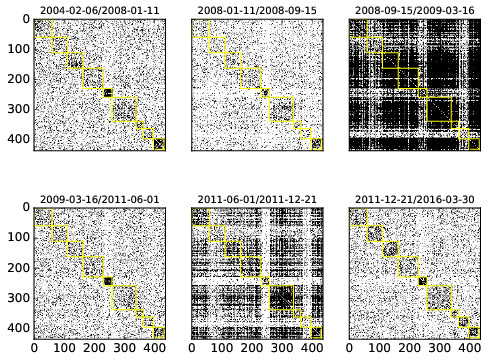}
\caption{$\hat{\theta}$ Adjaceny Matrices}
\label{fig:segments}
\end{figure}




\begin{figure}[h]
\centering     
\includegraphics[width=\textwidth]{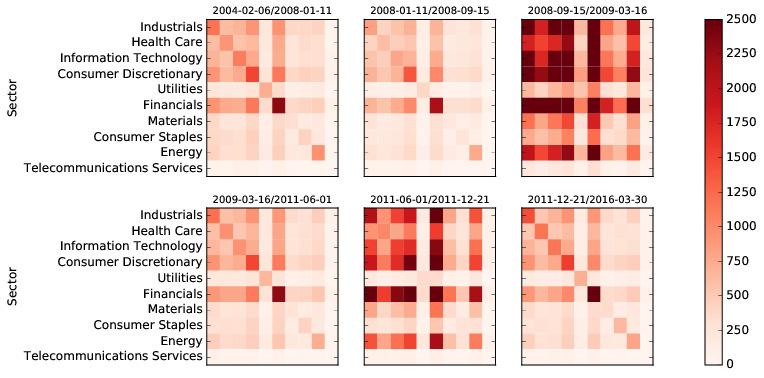}
\caption{Sector Edges}
\label{fig:indedges}
\end{figure}


From these figures we can get a sense of which sectors are most affected during times of crisis.  To expand upon this some, see Figure \ref{fig:fin_edges} for the edge count between each sector and the Financial sector for each $\hat{\theta}$ estimate.  We can see that during times of crisis, there is considerable connection between Industrials, Information Technology, Consumer Discretionary, and to a lesser extend Healthcare, and the Financial sector.  Consumer Staples, Utilities, and Materials appear to be more stable during these periods and do not experience as much correlation with Financials.  This might suggest that our method could be used as a tool to identify investment strategies that are likely to be resilient to periods of crisis in the market.

\begin{figure}[h]
\centering     
\includegraphics[width=\textwidth]{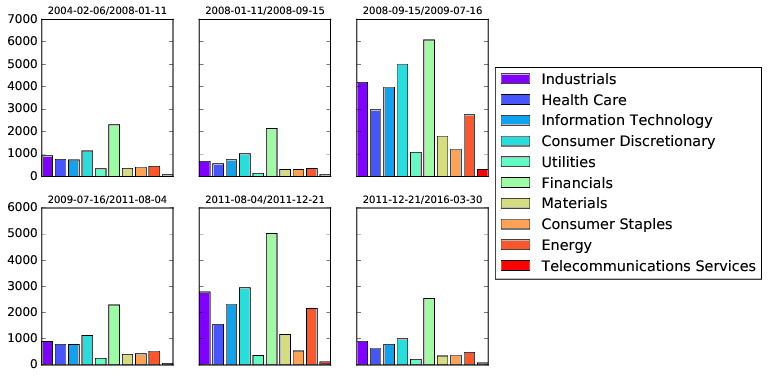}
\caption{Financial Sector Edges}
\label{fig:fin_edges}
\end{figure}

\section{proofs}\label{sec:proofs}
\subsection{Proof of Lemma \ref{lem1}}\label{sec:proof:lem1}
The proof is similar to the proof of Lemma 2 of \cite{atchade:etal:15}. We do the proof for $j=1$, the case $j=2$ being similar. Suppose that $\theta_1^{(k)}$ is non-singular. It is well known that
\begin{multline*}
\theta_1^{(k+1)} = \textsf{Argmin}_{u\in\M_p}\left[\pscal{\nabla g_{1,\tau^{(k)}}(\theta_1^{(k)})}{u-\theta_1^{(k)}} +\frac{1}{2\gamma}\normfro{u-\theta_1^{(k)}}^2 +\lambda_{1,\tau^{(k)}}\wp(u)\right].\end{multline*}
The optimality conditions of this problem implies that there exists $Z\in\rset^{p\times p}$, where $Z_{ij}\in [-1,1]$ for all $i,j$ such that
\begin{multline*}
\left(1+(1-\alpha)\lambda_{1,\tau^{(k)}}\gamma\right)\theta_1^{(k+1)} =\theta_1^{(k)} + \frac{\gamma\tau^{(k)}}{2T}\left(\theta_1^{(k)}\right)^{-1} -\gamma \left(\frac{\tau^{(k)}}{2T}S_1(\tau^{(k)}) + \alpha\lambda_{1,\tau^{(k)}}Z\right).\end{multline*}
Hence, if $\lambda_{\textsf{min}}(\theta_1^{(k)})\geq \textsf{b}_1$, and $b_1^2\geq \gamma\tau/(2T)$ (which holds true if $\gamma\leq 2\textsf{b}_1^2$), and using the fact that $\lambda_{\textsf{min}}(A+B) \geq \lambda_{\textsf{min}}(A) +\lambda_{\textsf{min}}(B)$, we get 
\[\lambda_{\textsf{min}}(\theta_1^{(k+1)})\geq \frac{1}{1+(1-\alpha)\bar\lambda_1\gamma}\left(\textsf{b}_1 +\frac{\gamma n_0}{2T} \frac{1}{\textsf{b}_1} - \gamma\mu_1\right)=\textsf{b}_1,\]
where the last equality follows from the fact that the chosen $\textsf{b}_1$ satisfies
\[(1-\alpha)\bar\lambda_1 \textsf{b}_1^2 +\mu_1 \textsf{b}_1 -\frac{n_0}{2T}=0.\]
Similarly, if $\lambda_{\textsf{max}}(\theta_1^{(k)})\leq \textsf{B}_1$, then
\[\lambda_{\textsf{max}}(\theta_1^{(k+1)})\leq \frac{1}{1+(1-\alpha)\underline{\lambda}_1\gamma}\left(\textsf{B}_1 +\frac{\gamma}{2} \frac{1}{\textsf{B}_1} + \gamma\mu_1\right)=\textsf{B}_1,\]
where the last equality follows from the fact that the chosen $\textsf{B}_1$ satisfies
\[(1-\alpha)\underline{\lambda}_1 \textsf{B}_1^2 -\mu_1 \textsf{B}_1 -\frac{1}{2}=0.\]

The argument that $\hat\theta_{j,\tau}\in \M_p^+(\textsf{b}_j,+\infty)$ is similar, and the details can be found for instance in the proof of Lemma 1 of \cite{atchade:etal:15}.
\begin{flushright}
$\square$
\end{flushright}

\subsection{Proof of Theorem \ref{thm1}}\label{sec:proof:thm1}
We will need the following lemma.
\begin{lemma}\label{lem2}
Set
\begin{multline*}
g(\theta)\eqdef -\log\det(\theta) +\textsf{Tr}(\theta S),\;\;\;\\
\mbox{ and }\;\; \phi(\theta)\eqdef g(\theta) +\lambda\left[\alpha\|\theta\|_1 +\frac{1-\alpha}{2}\normfro{\theta}^2\right],\;\;\theta\in\M_p^+,
\end{multline*}
for  some symmetric matrix $S$, $\alpha\in (0,1)$, and $\lambda>0$. Fix $0<b<B\leq \infty$.
\begin{enumerate}
\item For $\theta,\vartheta\in\M_p^+(b,B)$, we have
\begin{multline*}
g(\theta)+\pscal{\nabla g(\theta)}{\vartheta-\theta} +\frac{1}{2B^2}\normfro{\vartheta-\theta}^2\leq g(\vartheta)\\
\leq g(\theta)+\pscal{\nabla g(\theta)}{\vartheta-\theta} +\frac{1}{2b^2}\normfro{\vartheta-\theta}^2.\end{multline*}
More generally, If $\theta,\vartheta\in\M_p^+$, then
\[g(\vartheta)-g(\theta) -\pscal{\nabla g(\theta)}{\vartheta-\theta} \geq  \frac{\|\vartheta-\theta\|_{\textsf{F}}^2}{4\|\theta\|_2\left(\|\theta\|_2 + \frac{1}{2}\|\vartheta-\theta\|_{\textsf{F}}\right)}.\]
\item Let $\gamma\in (0,b^2]$, and $\theta,\bar\theta,\theta_0\in\M_p^+(b,B)$.  Suppose that 
\[\bar\theta = \Prox_{\gamma\lambda}\left(\theta-\gamma(S-\theta^{-1})\right),\]
then
\[2\gamma\left(\phi(\bar\theta)-\phi(\theta_0)\right) + \normfro{\bar\theta -\theta_0}^2\leq \left(1-\frac{\gamma}{B^2}\right)\normfro{\theta-\theta_0}^2.\]
\end{enumerate}
\end{lemma}
\begin{proof}
The first part of (1) is Lemma 12 of \cite{atchade:etal:15}, and Part (2) is Lemma 14 of  \cite{atchade:etal:15}. The second part of (1) can be proved along similar lines. For completeness we give the details below.

Take $\theta_0,\theta_1\in\M_p^+$.  By Taylor expansion we have
\[g(\theta_1)-g(\theta_0) -\pscal{\nabla g(\theta_0)}{\theta_1-\theta_0} = -\int_0^1 \pscal{(\theta_0+tH)^{-1}-\theta_0^{-1}}{H}\rmd t,\]
where $H\eqdef \theta_1-\theta_0$. We have $(\theta_0+tH)^{-1}-\theta_0^{-1} = -t\theta_0^{-1}H(\theta_0+tH)^{-1}$, which leads to
\[g(\theta_1)-g(\theta_0) -\pscal{\nabla g(\theta_0)}{\theta_1-\theta_0} = \int_0^1 \textsf{Tr}\left(\theta_0^{-1}H(\theta_0+tH)^{-1}H\right) t\rmd t.\]
If $\theta_0=\sum_{i=1}^p \rho_ju_ju_j'$ is the eigendecomposition of $\theta_0$, we see that
$\textsf{Tr}\left(\theta_0^{-1}H(\theta_0+tH)^{-1}H\right) =\sum_{j=1}^p \frac{1}{\rho_j}u_j'H(\theta_0+tH)^{-1}Hu_j$. Hence
\begin{multline*}
g(\theta_1)-g(\theta_0) -\pscal{\nabla g(\theta_0)}{\theta_1-\theta_0} \geq   \sum_{j=1}^p\|H u_j\|_2^2\int_0^1 \frac{t\rmd t}{\|\theta_0\|_2\left(\|\theta_0\|_2 + t\|H\|_{\textsf{F}}\right)}\\
\geq \frac{\sum_{j=1}^p\|H u_j\|_2^2}{4\|\theta_0\|_2\left(\|\theta_0\|_2 + \frac{1}{2}\|H\|_{\textsf{F}}\right)},\end{multline*}
and the result follows by noting that $\sum_{j=1}^p\|H u_j\|_2^2=\|H\|_{\textsf{F}}^2$.

\end{proof}
Set
\[\F(\tau,\theta_1,\theta_2) = g_{1,\tau}(\theta_1) + \lambda_{1,\tau}p(\theta) + g_{2,\tau}(\theta_2) + \lambda_{2,\tau}p(\theta_2),\]
$\underline{\F}=\F(\hat\tau,\hat\theta_{1,\hat\tau},\hat\theta_{1,\hat\tau})$ the value of Problem (\ref{est:2}), and $\F_k=\F(\tau^{(k)},\theta_1^{(k)},\theta_2^{(k)})-\underline{\F}$. 
\begin{lemma}\label{lem3}
Suppose that $\gamma\in(0,\textsf{b}_1^2\wedge \textsf{b}_2^2]$, and for $j=1,2$, $\theta_j^{(0)}\in\M_p^+(\textsf{b}_j,\textsf{B}_j)$. Then
 $\lim_k\normfro{\theta_1^{(k)}-\hat\theta_{1,\tau^{(k)}}}=0$, $\lim_k\normfro{\theta_2^{(k)}-\hat\theta_{2,\tau^{(k)}}}=0$. Furthermore the sequence $\{\F_k\}$ is non-increasing, and $\lim_k\F_k$ exists.
\end{lemma}
\begin{proof}
We know from Lemma \ref{lem1} that for $\gamma\in(0,\textsf{b}_1^2\wedge \textsf{b}_2^2]$, and $\theta_j^{(0)}\in\M_p^+(\textsf{b}_j,\textsf{B}_j)$, we have $\theta_j^{(k)}\in\M_p^+(\textsf{b}_j,\textsf{B}_j)$ for all $k\geq 0$, for $j=1,2$.  We have,
\begin{multline*}
\F_{k+1} -\F_k = \F(\tau^{(k+1)},\theta_1^{(k+1)},\theta_2^{(k+1)}) - \F(\tau^{(k)},\theta_1^{(k+1)},\theta_2^{(k+1)}) \\
+  \F(\tau^{(k)},\theta_1^{(k+1)},\theta_2^{(k+1)}) -\F(\tau^{(k)},\theta_1^{(k)},\theta_2^{(k)}).\end{multline*}
By definition, $\F(\tau^{(k+1)},\theta_1^{(k+1)},\theta_2^{(k+1)}) - \F(\tau^{(k)},\theta_1^{(k+1)},\theta_2^{(k+1)})\leq 0$, and by Lemma \ref{lem2}-Part(2), 
\begin{multline*}
\F(\tau^{(k)},\theta_1^{(k+1)},\theta_2^{(k+1)}) -\F(\tau^{(k)},\theta_1^{(k)},\theta_2^{(k)})\\
\leq -\frac{1}{2\gamma}\normfro{\theta_1^{(k+1)}-\theta_{1}^{(k)}}^2 - \frac{1}{2\gamma}\normfro{\theta_2^{(k+1)}-\theta_{2}^{(k)}}^2 \end{multline*}
It follows that 
\[ \F_{k+1} \leq \F_k-\frac{1}{2\gamma}\normfro{\theta_1^{(k+1)}-\theta_{1}^{(k)}}^2 - \frac{1}{2\gamma}\normfro{\theta_2^{(k+1)}-\theta_{2}^{(k)}}^2,\]
which implies that 
\begin{equation}\label{eq:cv:theta_k}
\lim_k \normfro{\theta_1^{(k+1)}-\theta_{1}^{(k)}}=0,\;\mbox{ and }\;\; \lim_k \normfro{\theta_2^{(k+1)}-\theta_{2}^{(k)}}=0.\end{equation}
It also implies that the sequence $\{\F_k\}$ is non-increasing and bounded from below by $0$. Hence converges.
Another application of Lemma \ref{lem2} gives
\begin{multline*}
2\gamma\left(\F(\tau^{(k)},\theta_1^{(k+1)},\theta_2^{(k+1)}) -\F(\tau^{(k)},\hat\theta_{1,\tau^{(k)}},\hat\theta_{2,\tau^{(k)}})\right)\\
  + \normfro{\theta_1^{(k+1)}-\hat\theta_{1,\tau^{(k)}}}^2 + \normfro{\theta_2^{(k+1)}-\hat\theta_{2,\tau^{(k)}}}^2 \\
\leq \left(1-\frac{\gamma}{\textsf{B}_1^2}\right)\normfro{\theta_1^{(k)}-\hat\theta_{1,\tau^{(k)}}}^2 + \left(1-\frac{\gamma}{\textsf{B}_2^2}\right)\normfro{\theta_2^{(k)}-\hat\theta_{2,\tau^{(k)}}}^2.\end{multline*}
And notice that $\F(\tau^{(k)},\theta_1^{(k+1)},\theta_2^{(k+1)}) -\F(\tau^{(k)},\hat\theta_{1,\tau^{(k)}},\hat\theta_{2,\tau^{(k)}})\geq 0$. Hence
\begin{multline*}
\normfro{\theta_1^{(k+1)}-\hat\theta_{1,\tau^{(k)}}}^2 + \normfro{\theta_2^{(k+1)}-\hat\theta_{2,\tau^{(k)}}}^2 \\
\leq \left(1-\frac{\gamma}{\textsf{B}_1^2}\right)\normfro{\theta_1^{(k)}-\hat\theta_{1,\tau^{(k)}}}^2 + \left(1-\frac{\gamma}{\textsf{B}_2^2}\right)\normfro{\theta_2^{(k)}-\hat\theta_{2,\tau^{(k)}}}^2,\end{multline*}
which can be written as
\begin{multline*}
\frac{\gamma}{\textsf{B}_1^2} \normfro{\theta_1^{(k)}-\hat\theta_{1,\tau^{(k)}}}^2 + \frac{\gamma}{\textsf{B}_2^2}\normfro{\theta_2^{(k)}-\hat\theta_{2,\tau^{(k)}}}^2 \leq \normfro{\theta_1^{(k+1)}-\theta_1^{(k)}}^2 +\normfro{\theta_2^{(k+1)}-\theta_2^{(k)}}^2 \\
-2\pscal{\theta_1^{(k+1)}-\theta_1^{(k)}}{\theta_1^{(k+1)}-\hat\theta_{1,\tau^{(k)}}} -2\pscal{\theta_2^{(k+1)}-\theta_2^{(k)}}{\theta_2^{(k+1)}-\hat\theta_{2,\tau^{(k)}}}.
\end{multline*}
Since $\{\theta_1^{(k)}\}$, $\{\theta_2^{(k)}\}$ $\{\hat\theta_{1,\tau^{(k)}}\}$, and $\{\hat\theta_{2,\tau^{(k)}}\}$ are bounded sequence, and given (\ref{eq:cv:theta_k}), letting $k\to\infty$, we conclude that
\[\lim_k \normfro{\theta_1^{(k)}-\hat\theta_{1,\tau^{(k)}}}=0,\;\;\mbox{ and }\;\; \lim_k \normfro{\theta_2^{(k)}-\hat\theta_{2,\tau^{(k)}}}=0.\]
\end{proof}

\begin{proof}[Proof of Theorem \ref{thm1}]
Let $\epsilon>0$ as in H\ref{H1}. By Lemma \ref{lem3}, there exist $k_0\geq 1$ such that for all $k\geq k_0$, $\normfro{\theta_1^{(k+1)}-\hat\theta_{1,\tau^{(k)}}}\leq \epsilon$, and $\normfro{\theta_2^{(k+1)}-\hat\theta_{2,\tau^{(k)}}}\leq \epsilon$. Since 
\[\tau^{(k+1)} = \textsf{Argmin}_{t\in\mathcal{T}}\; \mathcal{H}\left(t\vert \theta_1^{(k+1)},\theta_2^{(k+1)}\right),\] 
using H\ref{H1} we conclude that for all $k\geq k_0$,
\[\left|\tau^{(k+1)} -\tau_\star\right| \leq \kappa \left|\tau^{(k)} -\tau_\star\right| + c \leq \kappa^{k-k_0+1} \left|\tau^{(k_0)} -\tau_\star\right| +\frac{c}{1-\kappa},\]
which implies the stated result.
\end{proof}

\subsection{Proof of Theorem \ref{thm2}}\label{sec:proof:thm2}
We introduce some more notation. Given $M\in\rset^{p\times p}$ the sparsity structure of $M$ is the matrix $\delta\in\{0,1\}^{p\times p}$ such that $\delta_{jk} = \textbf{1}_{\{|M_{jk}|>0\}}$. In particular we will write $\delta_{\star,j}$ ($j=1,2$) to denote the sparsity structure of $\theta_{\star,j}$. Given matrices $A\in\rset^{p\times p}$, and $\delta\in\{0,1\}^{p\times p}$, we will use the notation $A_\delta$ (resp. $A_{\delta^c}$) to denote the component-wise product of $A$ and $\delta$ (resp $A$ and $1-\delta$). Given $j\in \{1,2\}$, we define
\begin{equation}\label{def:Cj}
\Cset_j\eqdef \left\{M\in\M_p:\; \|M_{\delta_{\star,j}^c}\|_1\leq 7\|M_{\delta_{\star,j}}\|_1.\right\}.\end{equation}

We will need the following deviation bound.

\begin{lemma}\label{lem:dev:bound}
Suppose that $X_i\stackrel{ind}{\sim}\textbf{N}(0,\theta_i^{-1})$, $i=1,\ldots,N$, where $\theta_i\in\M_p^+$. We set $\Sigma_i \eqdef \theta_i^{-1}$, and define
\[\underline{\kappa}_i(2) \eqdef\inf\left\{u'\Sigma_iu,\;\|u\|_2=1,\;\|u\|_0\leq 2\right\},\;\;\bar{\kappa}_i(2) \eqdef\sup\left\{u'\Sigma_iu,\;\|u\|_2=1,\;\|u\|_0\leq 2\right\},\]
and suppose that $\underline{\kappa}_i(2)>0$ for $i=1,\ldots,N$. Set $G_N\eqdef N^{-1}\sum_{i=1}^N (X_iX_i'-\theta_i^{-1})$. Then for $0<\delta\leq 2\left(\frac{\min_k \underline{\kappa}_k(2)}{\max_k\bar \kappa_k(2)}\right)^2$, we have
\[\PP\left(\|G_N\|_\infty>\left(\max_k\bar \kappa_k(2)\right)\delta\right)\leq 4p^2 e^{-\frac{N \delta^2}{4}}.\]
\end{lemma}
\begin{proof}
The proof is similar to the proof of Lemma 1 of \cite{ravikumaretal10}, which itself builds on \cite{bickeletlevina08a}. For $1\leq i,j\leq p$, arbitrary, set $Z^{(k)}_{ij}=X_{k,i}X_{k,j}$, and $\sigma_{ij}^{(k)}=\Sigma_{k,ij}$, so that the $(i,j)$-th component of $G_N$ is $N^{-1}\sum_{k=1}^N(Z^{(k)}_{ij}-\sigma^{(k)}_{ij})$. Suppose that $i\neq j$. The case $i=j$ is simpler. It is easy to check that
\begin{multline*}
\sum_{k=1}^N \left[Z^{(k)}_{ij}-\sigma^{(k)}_{ij}\right]=\frac{1}{4}\sum_{k=1}^N \left[(X_{k,i}+X_{k,j})^2-\sigma^{(k)}_{ii}-\sigma^{(k)}_{jj}-2\sigma^{(k)}_{ij}\right] \\
-\frac{1}{4}\sum_{k=1}^N \left[(X_{k,i}-X_{k,j})^2-\sigma^{(k)}_{ii}-\sigma^{(k)}_{jj}+2\sigma^{(k)}_{ij}\right].\end{multline*}
Notice that $X_{k,i}+X_{k,j}\sim\textbf{N}(0,\sigma^{(k)}_{ii}+\sigma^{(k)}_{jj}+2\sigma^{(k)}_{ij})$, and $X_{k,i}-X_{k,j}\sim\textbf{N}(0,\sigma^{(k)}_{ii}+\sigma^{(k)}_{jj}-2\sigma^{(k)}_{ij})$. It follows that for all $x\geq 0$,
\begin{multline*}
\PP\left[\left|\sum_{k=1}^N \left[Z^{(k)}_{ij}-\sigma^{(k)}_{ij}\right]\right|> x\right]\leq \PP\left[\left|\sum_{k=1}^Na_{ij}^{(k)}(W_k-1)\right|>2x\right] \\
+\PP\left[\left|\sum_{k=1}^Nb_{ij}^{(k)}(W_k-1)\right|>2x\right],
\end{multline*}
where $W_{1:N}\stackrel{i.i.d.}{\sim}\chi_1^2$, $a_{ij}^{(k)} =\sigma^{(k)}_{ii}+\sigma^{(k)}_{jj}+2\sigma^{(k)}_{ij}$, and $b_{ij}^{(k)} =\sigma^{(k)}_{ii}+\sigma^{(k)}_{jj}-2\sigma^{(k)}_{ij}$. For any $x\geq 0$ and a sequence $a=(a_1,\ldots,a_N)$ of positive numbers, with $|a|_\infty =\max_i|a_i|$, $|a|_2=\sqrt{\sum_ia_i^2}$, we write
\[2x = 2|a|_2\left(\frac{x}{2|a|_2}\right) + 2|a|_\infty\left(\frac{4|a|_2^2}{2x|a|_\infty}\right)\left(\frac{x}{2|a|_2}\right)^2.\]
Therefore if $2x|a|_\infty\leq 4|a|_2^2$, we can apply Lemma 1 of \cite{laurent:massart:00} to conclude that
\[\PP\left(\left|\sum_{k=1}^N a_k(W_k-1)\right|\geq 2x\right) \leq 2e^{-\frac{x^2}{4|a|_2^2}}.\]
In particular, we can apply the above bound with $x=|a|_\infty N\delta$ for $\delta\in (0,\frac{2\min_ja_i^2}{\max_i a_i^2}]$ to get that 
\[\PP\left(\left|\sum_{k=1}^N a_k(W_k-1)\right|\geq 2|a|_\infty N\delta\right) \leq 2e^{-\frac{N\delta^2}{4}}.\]

In the particular case above, $a_{ij}^{(k)} = \sigma^{(k)}_{ii} + \sigma^{(k)}_{jj} +2\sigma^{(k)}_{ij} =u'\Sigma^{(k)}u$, where $u_i=u_j=1$, and $u_r=0$ for $r\notin\{i,j\}$. And 
\[\frac{\min_k u'\Sigma^{(k)} u}{\max_k u'\Sigma^{(k)} u} \geq \frac{\min_k \underline{\kappa}_k(2)}{\max_k \bar \kappa(2)}.\]
A similar bound holds for $b_{ij}^{(k)}$.  The lemma follows from a standard union-sum argument.

\end{proof}

The following event plays an important role in the analysis.
\begin{equation}\label{def:en}
\e_n\eqdef \bigcap_{\tau\in\mathcal{T}} \left\{\frac{1}{\lambda_{1,\tau}}\|\nabla g_{1,\tau}(\theta_{\star,1})\|_\infty\leq \frac{\alpha}{2},\;\mbox{ and }\; \frac{1}{\lambda_{2,\tau}}\|\nabla g_{2,\tau}(\theta_{\star,2})\|_\infty\leq \frac{\alpha}{2}\right\},\end{equation}
\begin{lemma}\label{lem1:thm2}
Under the assumptions of the theorem
\[\PP(\e_n)\geq 1-\frac{8}{pT}.\]
\end{lemma}
\begin{proof}
We have
\[\PP(\e_n^c)\leq \PP\left(\max_{\tau\in\mathcal{T}}\frac{1}{\lambda_{1,\tau}}\|\nabla g_{1,\tau}(\theta_{\star,1})\|_\infty>\frac{\alpha}{2}\right) + \PP\left(\max_{\tau\in\mathcal{T}}\frac{1}{\lambda_{2,\tau}}\|\nabla g_{2,\tau}(\theta_{\star,2})\|_\infty>\frac{\alpha}{2}\right).\]
We show how to bound the first term. A similar bound follows for $g_{2,\tau}$ by working on the reversed sequence $X^{(T)},\ldots,X^{(1)}$. We have $\nabla g_{1,\tau}(\theta) = \frac{\tau}{2T}(S_1(\tau)-\theta^{-1})$. Setting $U^{(t)}\eqdef X^{(t)}(X^{(t)})' - \PE\left(X^{(t)}(X^{(t)})'\right)$, we can write
\[\nabla g_{1,\tau}(\theta_{\star,1}) =\frac{1}{2T}\sum_{t=1}^\tau U^{(t)}  + \frac{(\tau-\tau_\star)_+}{2T}(\theta_{\star,2}^{-1}-\theta_{\star,1}^{-1}),\]
where $a_+\eqdef\max(a,0)$. Hence by a standard union-bound argument, 
\begin{multline*}
\PP\left(\max_{\tau\in\mathcal{T}}\frac{1}{\lambda_{1,\tau}}\|\nabla g_{1,\tau}(\theta_{\star,1})\|_\infty>\frac{\alpha}{2}\right) \\
\leq \sum_{\tau\in\mathcal{T}}\PP\left(\left\|\sum_{t=1}^\tau U^{(t)}\right\|_\infty > \alpha\lambda_{1,\tau} T -(\tau-\tau_\star)_+\|\theta_{\star,2}^{-1}-\theta_{\star,1}^{-1}\|_\infty\right).\end{multline*}
Given the choice of $\lambda_{1,\tau}$ in (\ref{eq:lambda}), $\alpha\lambda_{1,\tau} T/2 = 2\sqrt{3}\bar\kappa\sqrt{\tau\log(pT)} \geq (\tau-\tau_\star)_+\|\theta_{\star,2}^{-1}-\theta_{\star,1}^{-1}\|_\infty$, by assumption (\ref{eq:cond:T:2}). In view of (\ref{eq:cond:T:1}) we can apply Lemma \ref{lem:dev:bound} to deduce that
\begin{eqnarray*}
\PP\left(\max_{\tau\in\mathcal{T}}\frac{1}{\lambda_{1,\tau}}\|\nabla g_{1,\tau}(\theta_{\star,1})\|_\infty>\frac{\alpha}{2}\right) & \leq & \sum_{\tau\in\mathcal{T}}\PP\left(\left\|\frac{1}{\tau}\sum_{t=1}^\tau U^{(t)}\right\|_\infty > \frac{\alpha\lambda_{1,\tau} T}{2\tau}\right)\\
& \leq & 4Tp^2 e^{-\frac{\tau}{4}\left(\frac{\alpha\lambda_{1,\tau}T}{2\tau\bar \kappa}\right)^2}\\
& \leq & 4\exp\left(2\log(pT)-3\log(pT)\right) \leq  \frac{4}{pT}.
\end{eqnarray*}

\end{proof}
\begin{lemma}\label{lem2:thm2}
Under the assumptions of the theorem, and on the event $\e_n$, we have
\[\normfro{\hat\theta_{1,\tau} -\theta_{\star,1}} \leq A \bar\kappa \|\theta_{\star,1}\|_2^2\sqrt{\frac{s_1\log(pT)}{\tau}},\]
and
\[\normfro{\hat\theta_{2,\tau} -\theta_{\star,2}} \leq A \bar\kappa \|\theta_{\star,2}\|_2^2\sqrt{\frac{s_2\log(pT)}{T-\tau}},\]
for all $\tau\in\mathcal{T}$, where $A$ is an absolute constant that can be taken as $A=16\times 20\times \sqrt{48}$.
\end{lemma}
\begin{proof}
Fix $j\in \{1,2\}$, and $\tau\in\mathcal{T}$. Set $\bar g_{j,\tau}(\theta) \eqdef g_{j,\tau}(\theta) +(1-\alpha)\lambda_{j,\tau}\normfro{\theta}/2$, and  recall that $\phi_{j,\tau}(\theta) \eqdef  g_{j,\tau}(\theta) + \lambda_{j,\tau}\wp(\theta)$. Hence $\phi_{j,\tau}(\theta) = \bar g_{j,\tau}(\theta) +\alpha\lambda_{j,\tau}\|\theta\|_1$. By a very standard argument that can be found for instance in \cite{negahbanetal10}, it is known that on the event $\e_n$, and if $\alpha$ satisfies (\ref{eq:cond:alpha}) then we have $\hat\theta_{j,\tau} -\theta_{\star,j}\in \Cset_j$, where the cones $\Cset_j$ are as defined in (\ref{def:Cj}). We write
\begin{eqnarray*}
\phi_{j,\tau}(\hat\theta_{j,\tau})- \phi_{j,\tau}(\theta_{\star,j})
 & = & \pscal{\nabla g_{j,\tau}(\theta_{\star,j}) +(1-\alpha)\lambda_{j,\tau}\theta_{\star,j}}{\hat\theta_{j,\tau} - \theta_{\star,j}} \\
&& + \bar g_{j,\tau}(\hat\theta_{j,\tau}) - \bar g_{j,\tau}(\theta_{\star,j}) -\pscal{\nabla \bar g_{j,\tau}(\theta_{\star,j})}{\hat\theta_{j,\tau} - \theta_{\star,j}}\\
&& + \alpha\lambda_{j,\tau}\left(\|\hat\theta_{j,\tau}\|_1 - \|\theta_{\star,j}\|_1\right).
\end{eqnarray*}
On $\e_n$,  $\hat\theta_{j,\tau} -\theta_{\star,j}\in \Cset_j$. Therefore
\[\alpha\lambda_{j,\tau}\left|\|\hat\theta_{j,\tau}\|_1 - \|\theta_{\star,j}\|_1\right| \leq \alpha\lambda_{j,\tau}\left\|\hat\theta_{j,\tau}-\theta_{\star,j}\right\|_1 \leq 8\alpha\lambda_{j,\tau}\sqrt{s_j}\normfro{\hat\theta_{j,\tau}-\theta_{\star,j}},\]
and
\begin{multline*}
\left|\pscal{\nabla g_{j,\tau}(\theta_{\star,j}) +(1-\alpha)\lambda_{j,\tau}\theta_{\star,j}}{\hat\theta_{j,\tau} - \theta_{\star,j}} \right|\\
\leq \frac{\lambda_{j,\tau}}{2}\left(\alpha +2(1-\alpha)\|\theta_{\star,j}\|_\infty\right)\left\|\hat\theta_{j,\tau}-\theta_{\star,j}\right\|_1\\
\leq 4\lambda_{j,\tau}\left(\alpha +2(1-\alpha)\|\theta_{\star,j}\|_\infty\right)\sqrt{s_j}\normfro{\hat\theta_{j,\tau}-\theta_{\star,j}}.
\end{multline*}
Suppose $j=1$. The case $j=2$ is similar. We then set $\Delta_{1,\tau}\eqdef \hat\theta_{1,\tau}-\theta_{\star,1}$, and use the second part of Lemma \ref{lem2}~(1) to deduce that
\begin{multline*}
\bar g_{1,\tau}(\hat\theta_{1,\tau}) - \bar g_{1,\tau}(\theta_{\star,1}) -\pscal{\nabla \bar g_{1,\tau}(\theta_{\star,1})}{\hat\theta_{1,\tau} - \theta_{\star,1}} \\
\geq g_{1,\tau}(\hat\theta_{1,\tau}) - g_{1,\tau}(\theta_{\star,1}) -\pscal{\nabla g_{1,\tau}(\theta_{\star,1})}{\hat\theta_{1,\tau} - \theta_{\star,1}}\\
\geq \frac{\tau}{2T} \frac{\|\Delta_{1,\tau}\|_{\textsf{F}}^2}{2\|\theta_{\star,1}\|_2\left(2\|\theta_{\star,1}\|_2+ \|\Delta_{1,\tau}\|_{\textsf{F}}\right)}.
\end{multline*}
Set $c_1= \frac{\tau}{4T\|\theta_{\star,1}\|_2^2}$, $c_2 =4\lambda_{1,\tau}\sqrt{s_1}\left(3\alpha + 2(1-\alpha)\|\theta_{\star,1}\|_\infty\right)$. Since $\phi_{1,\tau}(\hat\theta_{1,\tau})- \phi_{1,\tau}(\theta_{\star,1})\leq 0$, the above derivation shows that on the event $\e_n$,
\[\frac{c_1\normfro{\Delta_{1,\tau}}^2}{2+\frac{1}{\|\theta_{\star,1}\|_2}\normfro{\Delta_{1,\tau}}} -c_2\normfro{\Delta_{1,\tau}} \leq 0,\]
Under the assumption that $c_1\geq 2c_2/\|\theta_{\star,1}\|_2$ (which we impose in (\ref{eq:cond:T:1})), this implies that
\[\normfro{\Delta_{1,\tau}} \leq \frac{4c_2}{c_1} \leq A \bar\kappa \|\theta_{\star,1}\|_2^2\sqrt{\frac{s_1\log(pT)}{\tau}},\]
where $A = 16\times 20\times \sqrt{48}$, as claimed.
\end{proof}

\begin{proof}[Proof of Theorem \ref{thm2}]
For $\tau\in\mathcal{T}$, let
\[r_{1,\tau} \eqdef A \bar\kappa \|\theta_{\star,1}\|_2^2\sqrt{\frac{s_1\log(pT)}{\tau}},\;\;\;r_{2,\tau} \eqdef A \bar\kappa \|\theta_{\star,2}\|_2^2\sqrt{\frac{s_2\log(pT)}{T-\tau}},\]
be  the convergence rates obtained in Lemma \ref{lem2:thm2}. Let $\epsilon>0$ be such that 
\[\epsilon \leq \min_{\tau\in\mathcal{T}} (r_{1,\tau}\wedge r_{1,\tau}).\]

 For $j=1,2$, let $\theta_j\in\M_p^+$ be such that $\|\theta_j-\hat\theta_{\tau,j}\|_1\leq \epsilon$. Set $\check\tau =\textsf{Argmin}_{t\in\mathcal{T}}\; \mathcal{H}(t\vert \theta_1,\theta_2)$, where $\mathcal{H}$ is as defined in (\ref{def:H}). Set 
 \[C_0 =  \min\left[\frac{\|\theta_{\star,2}-\theta_{\star,1}\|_{\textsf{F}}^4}{128 B^4\|\theta_{\star,2}-\theta_{\star,1}\|_1^2},\left(\frac{\underline{\kappa}}{\bar\kappa}\right)^4\right].\]
 We will show below that 
 \begin{equation}\label{thm2:eq0}
 \PP\left(|\check\tau-\tau_\star|>\frac{4\log(p)}{C_0}\right)\leq \frac{8}{pT} +\frac{4}{p^2\left(1-e^{-C_0}\right)}.\end{equation}
 This implies that with probability at least $1-\frac{8}{pT} - \frac{4}{p^2\left(1-e^{-C_0}\right)}$, Assumption H\ref{H1} holds (with $\epsilon\leftarrow \epsilon/\sqrt{p}$, $\kappa=0$, and $c=(4/C_0)\log(p)$). The theorem then follows by applying Theorem \ref{thm1}.

Given $\theta_j\in\M_p^+$ be such that $\|\theta_j-\hat\theta_{\tau,j}\|_1\leq \epsilon$, we will now show that (\ref{thm2:eq0}) holds.  We shall bound $\PP(\check\tau>\tau_\star+ \delta)$, $\delta=(4/C_0)\log(p)$. The bound on $\PP(\check\tau<\tau_\star-\delta)$ follows similarly by working with the reversed sequence $X^{(T)},\ldots,X^{(1)}$.

Note that $\theta_j$ can be written as
\begin{equation}\label{decomp:theta}
\theta_j = (\theta_j-\hat\theta_{\tau,j}) +(\hat\theta_{\tau,j}-\theta_{\star,j}) + \theta_{\star,j}.\end{equation}
This implies that on $\e_n$, for $\epsilon\leq r_{j,\tau}$, and $r_{j,\tau}\leq\min\left(\frac{ \lambda_{\textsf{min}}(\theta_{\star,j})}{4},\frac{\|\theta_{\star,j}\|_\infty}{2},\frac{\|\theta_{\star,j}\|_1}{1+8s_j^{1/2}}\right)$, we have
\begin{multline}\label{bound:spec:theta}
\lambda_{\textsf{min}}(\theta_j) \geq \frac{1}{2}\lambda_{\textsf{min}}(\theta_{\star,j}),\;\;\;\; \lambda_{\textsf{max}}(\theta_j) \leq 2\lambda_{\textsf{max}}(\theta_{\star,j}),\\
\|\theta_j\|_\infty \leq 2\|\theta_{\star,j}\|_\infty,\;\;\;\mbox{ and }\;\;\; \|\theta_j\|_1 \leq 2\|\theta_{\star,j}\|_1.\end{multline} 
Using the event $\e_n$  introduced in (\ref{def:en}), we have
\begin{multline}\label{thm2:eq1}
\PP\left(\check\tau>\tau_\star + \delta\right)  \leq    \PP(\e_n^c) + \sum_{j\geq 0:\; \tau_\star+\delta+j\in\mathcal{T}}\PP\left(\e_n,\; \check\tau = \tau_\star +\delta+j\right) \\
 \leq  \PP(\e_n^c) + \sum_{j\geq 0:\; \tau_\star+\delta+j\in\mathcal{T}}\PP\left(\e_n,\; \phi_{1,\tau_\star +\delta+j}(\theta_1) +\phi_{2,\tau_\star+\delta+j}(\theta_2)\leq \phi_{1,\tau_\star}(\theta_1) +\phi_{2,\tau_\star}(\theta_2)\right),
\end{multline}
where $\phi_{j,\tau}(\theta) \eqdef  g_{j,\tau}(\theta) + \lambda_{j,\tau}\wp(\theta)$. First we are going to bound the probability 
\[\PP\left(\e_n,\; \phi_{1,\tau}(\theta_1) +\phi_{2,\tau}(\theta_2)\leq \phi_{1,\tau_\star}(\theta_1) +\phi_{2,\tau_\star}(\theta_2)\right),\]
for some arbitrary $\tau\in\mathcal{T}$, $\tau>\tau_\star$.  A simple calculation shows that 
\begin{multline*}
\frac{2T}{\tau-\tau_\star}\left[\phi_{1,\tau}(\theta_1) +\phi_{2,,\tau}(\theta_2) -\phi_{1,\tau_\star}(\theta_1) -\phi_{2,\tau_\star}(\theta_2)\right] = -\log\det(\theta_1)+\log\det(\theta_2) \\
+\pscal{\theta_1-\theta_2}{\theta_{\star,2}^{-1}}+\pscal{\theta_1-\theta_2}{\frac{1}{\tau-\tau_\star}\sum_{t=\tau_\star+1}^\tau \left(X^{(t)}X^{(t)'}-\theta_{\star,2}^{-1}\right)}\\
+2T\left(\frac{\lambda_{1,\tau}-\lambda_{1,\tau_\star}}{\tau-\tau_\star}\right)\left(\frac{1-\alpha}{2}\|\theta_1\|_{\textsf{F}}^2 +\alpha\|\theta_1\|_1\right) \\
+2T\left(\frac{\lambda_{2,\tau}-\lambda_{2,\tau_\star}}{\tau-\tau_\star}\right)\left(\frac{1-\alpha}{2}\|\theta_2\|_{\textsf{F}}^2 +\alpha\|\theta_2\|_1\right).
\end{multline*}
We have $2T\left(\frac{\lambda_{1,\tau}-\lambda_{1,\tau_\star}}{\tau-\tau_\star}\right)\left(\frac{1-\alpha}{2}\|\theta_1\|_{\textsf{F}}^2 +\alpha\|\theta_1\|_1\right)\geq 0$, and
\[2T\left|\frac{\lambda_{2,\tau}-\lambda_{2,\tau_\star}}{\tau-\tau_\star}\right| \leq \frac{\bar\kappa}{\alpha} \sqrt{\frac{48\log(pT)}{T-\tau}} = \frac{c_0r_{2,\tau}}{\alpha s_2^{1/2}\|\theta_{\star,2}\|_2^2},\]
for some absolute constant $c_0$. Using the infinity-norm and $1$-norm bounds in (\ref{bound:spec:theta}) together with (\ref{eq:cond:alpha}), we have 
\[\frac{1-\alpha}{2}\|\theta_2\|_{\textsf{F}}^2 +\alpha\|\theta_2\|_1 =\alpha\left[\frac{1-\alpha}{2\alpha}\|\theta_2\|_{\infty} +1\right]\|\theta_2\|_1 \leq 4\alpha \|\theta_{\star,2}\|_1,\]
and it follows that
\[2T\left|\frac{\lambda_{2,\tau}-\lambda_{2,\tau_\star}}{\tau-\tau_\star}\right|\left(\frac{1-\alpha}{2}\|\theta_2\|_{\textsf{F}}^2 +\alpha\|\theta_2\|_1\right) \leq C_\tau\eqdef \left(\frac{4c_0\|\theta_{\star,2}\|_1}{s_2^{1/2}\|\theta_{\star,2}\|_2^2}\right)r_{2,\tau}.\]
Set
\[b \eqdef \min\left(\lambda_{\textsf{min}}(\theta_{\star,1}),\lambda_{\textsf{min}}(\theta_{\star,2})\right),\;\; B \eqdef \max\left(\|\theta_{\star,1}\|_2,\|\theta_{\star,2}\|_2\right).\] 
By the strong convexity of $\log\det$ (Lemma \ref{lem2}~Part(1)) we have:
\begin{multline*}
-\log\det(\theta_1)+\log\det(\theta_2) +\pscal{\theta_1-\theta_2}{\theta_{\star,2}^{-1}} \\
\geq \pscal{\theta_{\star,2}^{-1} -\theta_2^{-1}}{\theta_1-\theta_2} +\frac{1}{2B^2}\|\theta_1-\theta_2\|_{\textsf{F}}^2.\end{multline*}
Since $\theta_{\star,2}^{-1} -\theta_2^{-1} = \theta_{\star,2}^{-1}(\theta_2-\theta_{\star,2})\theta_2^{-1}$,  and using the fact that $\|AB\|_{\textsf{F}} \leq \|A\|_2\|B\|_{\textsf{F}}$, we have that on $\e_n$,
\begin{multline*}\left|\pscal{\theta_{\star,2}^{-1} -\theta_2^{-1}}{\theta_1-\theta_2}\right| \leq 2r_{2,\tau}\|\theta_{\star,2}^{-1}\|_{2} \|\theta_2^{-1}\|_2\|\theta_2-\theta_1\|_{\textsf{F}}\leq 4r_{2,\tau}\|\theta_{\star,2}^{-1}\|_{2}^2\|\theta_2-\theta_1\|_{\textsf{F}}. \end{multline*}
We conclude that on $\e_n$,
\begin{multline*}
\frac{2T}{\tau-\tau_\star}\left[\phi_{1,\tau}(\theta_1) +\phi_{2,\tau}(\theta_2) -\phi_{1,\tau_\star}(\theta_1) -\phi_{2,\tau_\star}(\theta_2)\right] \geq \\
\pscal{\theta_1-\theta_2}{\frac{1}{\tau-\tau_\star}\sum_{t=\tau_\star+1}^\tau \left(X^{(t)}X^{(t)'}-\theta_{\star,2}^{-1}\right)} \\
-C_\tau -4r_{2,\tau}\|\theta_{\star,2}^{-1}\|_2^2 \|\theta_2-\theta_1\|_{\textsf{F}} +\frac{1}{2B^2}\|\theta_1-\theta_2\|_{\textsf{F}}^2.
\end{multline*}
Under the assumption (\ref{eq:tech:cond:thm2}) imposed on $r_{j,\tau}$ and for $\epsilon\leq r_{1,\tau}\wedge r_{2,\tau}$, it can be shown that on $\e_n$, and for $\|\theta_{\star,2} -\theta_{\star,1}\|_{\textsf{F}} \geq \frac{8c_0\|\theta_{\star,2}\|_1}{s_2^{1/2} \|\theta_{\star,2}\|_2^2\|\theta_{\star,2}^{-1}\|_2^2}$, we have 
\begin{equation}\label{eq:quad}
-C_\tau -2\left(\epsilon + r_{2,\tau}\right)\|\theta_{\star,2}^{-1}\|_{2}^2 \|\theta_2-\theta_1\|_{\textsf{F}} +\frac{1}{4B^2}\|\theta_1-\theta_2\|_{\textsf{F}}^2\geq 0.\end{equation}
To see this, note that (\ref{eq:quad}) holds if $\|\theta_2-\theta_1\|_{\textsf{F}} \geq 8B^2 r_{2,\tau}\|\theta_{\star,2}^{-1}\|_2^2 +2B\sqrt{C_\tau + 16 B^2\|\theta_{\star,2}^{-1}\|_2^4 r_{2,\tau}^2}$. Then it can be checked that if $r_{2,\tau} \leq \frac{c_0\|\theta_{\star,2}\|_1}{16B^2s_2^{1/2}\|\theta_{\star,2}\|_2^2\|\theta_{\star,2}^{-1}\|_2^4}$, then 
\[8B^2\|\theta_{\star,2}^{-1}\|_2^2 r_{2,\tau}\leq \frac{C_\tau}{2\|\theta_{\star,2}^{-1}\|_2^2 r_{2,\tau}},\;\;\;\mbox{ and }\;\;\; 4B\sqrt{C_\tau}\leq \frac{C_\tau}{2\|\theta_{\star,2}^{-1}\|_2^2 r_{2,\tau}}.\]
Therefore, (\ref{eq:quad}) holds if 
\[ \|\theta_2-\theta_1\|_{\textsf{F}} \geq \frac{C_\tau}{\|\theta_{\star,2}^{-1}\|_2^2 r_{2,\tau}} = \frac{4c_0\|\theta_{\star,2}\|_1}{s_2^{1/2}\|\theta_{\star,2}\|_2^2\|\theta_{\star,2}^{-1}\|_2^2}.\]
Now we write
\[\theta_2-\theta_1 = (\theta_2-\hat\theta_{\tau,2}) + (\hat\theta_{\tau,2}-\theta_{\star,2}) + (\theta_{\star,2}-\theta_{\star,1}) + (\theta_{\star,1}-\hat\theta_{\tau,1}) +(\hat\theta_{\tau,1} -\theta_1),\]
and use the fact that $\epsilon\leq r_{1,\tau}\wedge r_{2,\tau}$, and $r_{j,\tau}\leq \|\theta_{\star,2}-\theta_{\star,1}\|_{\textsf{F}}/8$ to deduce that on $\e_n$, $\|\theta_2-\theta_1\|_{\textsf{F}} \geq  \|\theta_{\star,2}-\theta_{\star,1}\|_{\textsf{F}}/2$, and this completes the proof of the claim.

It follows from the above that 
\begin{multline}\label{thm2:eq2}
\PP\left(\e_n;\phi_{1,\tau}(\theta_1) +\phi_{2,\tau}(\theta_2) -\phi_{1,\tau_\star}(\theta_1) -\phi_{2,\tau_\star}(\theta_2) \leq 0\right)\\
\leq \PP\left(\left\|\frac{1}{\tau-\tau_\star}\sum_{t=\tau_\star+1}^\tau \left(X^{(t)}X^{(t)'}-\theta_{\star,2}^{-1}\right)\right\|_\infty>\frac{\|\theta_2-\theta_1\|_{\textsf{F}}^2}{4B^2\|\theta_2-\theta_1\|_1}\right).\end{multline}
Proceeding as above, it is easy to see that if $\epsilon\leq r_{1,\tau}\wedge r_{2,\tau}$, and $r_{j,\tau}\leq \frac{\|\theta_{\star,2}-\theta_{\star,1}\|_{\textsf{F}}}{2(1+8s^{1/2})}$, then
\[\frac{\|\theta_2-\theta_1\|_{\textsf{F}}^2}{4B^2\|\theta_2-\theta_1\|_1} \geq \frac{\|\theta_{\star,2}-\theta_{\star,1}\|_{\textsf{F}}^2}{32B^2\|\theta_{\star,2}-\theta_{\star,1}\|_1}.\]
Using this, and by Lemma \ref{lem1:thm2}, it follows that the probability on the right-hand side of (\ref{thm2:eq2}) is upper-bounded by 
\[4p^2\exp\left(-(\tau-\tau_\star)\min\left[\frac{\|\theta_{\star,2}-\theta_{\star,1}\|_{\textsf{F}}^4}{128 B^4\|\theta_{\star,2}-\theta_{\star,1}\|_1^2},\left(\frac{\underline{\kappa}}{\bar\kappa}\right)^4\right]\right).\]
We apply this to (\ref{thm2:eq1}) to get:
\[\PP(\check\tau>\tau_\star +\delta) \leq \PP(\e_n^c) + \sum_{j\geq 0} 4p^2e^{-C_0(\delta+j)}\leq\frac{8}{pT} +\frac{4}{p^2(1-e^{-C_0})},\]
where $C_0 =  \min\left[\frac{\|\theta_{\star,2}-\theta_{\star,1}\|_{\textsf{F}}^4}{128 B^4\|\theta_{\star,2}-\theta_{\star,1}\|_1^2},\left(\frac{\underline{\kappa}}{\bar\kappa}\right)^4\right]$, and by taking $\delta = 4\log(p)/C_0$. This completes the proof.

\end{proof}

\bibliographystyle{ims}
\bibliography{biblio_graph,biblio_optim,biblio_data}

\end{document}

%% file: header_new.tex
\usepackage{amsmath,amssymb,amsfonts,amsthm,enumerate,natbib,color,ifthen,hyperref}
\usepackage{xargs}
\usepackage{accents}
\usepackage[textwidth=4cm, textsize=footnotesize]{todonotes}
\usepackage[latin1]{inputenc}

\usepackage{float}
\usepackage{graphicx}
\usepackage{subfigure}

\usepackage{aliascnt,bbm}
\def\rset{\mathbb R}
\def\zset{\mathbb Z}

\def\eqsp{\;}

 \newcommand{\pscal}[2]{\left\langle#1,#2\right\rangle}

\newcommand{\eqdef}{\ensuremath{\stackrel{\mathrm{def}}{=}}}

\def\F{\mathcal{F}} 
\def\e{\mathcal{E}}

\def\M{\mathcal{M}}

\def\H{\mathcal{H}}

\newcommandx\sequence[3][2=t,3=\zset]
{\ifthenelse{\equal{#3}{}}{\ensuremath{\{ #1_{#2}\}}}{\ensuremath{\{ #1_{#2}, \eqsp #2 \in #3 \}}}}

\def\PP{\mathbb{P}} 
\newcommand{\CPP}[3][]
{\ifthenelse{\equal{#1}{}}{{\mathbb P}\left(\left. #2 \, \right| #3 \right)}{{\mathbb P}_{#1}\left(\left. #2 \, \right | #3 \right)}}
\def\PE{\mathbb{E}} 
\newcommand{\CPE}[3][]
{\ifthenelse{\equal{#1}{}}{{\mathbb E}\left[\left. #2 \, \right| #3 \right]}{{\mathbb E}_{#1}\left[\left. #2 \, \right | #3 \right]}}

\def\Cset{\mathcal{C}} 

\newcommand{\normfro}[1]{\left\Vert#1\right\Vert_{\textsf{F}}}

\def\Prox{\operatorname{Prox}}

\usepackage{bm}

\theoremstyle{plain}
\newtheorem{theorem}{Theorem}
\newtheorem{assumption}{H\hspace{-3pt}}

\newaliascnt{proposition}{theorem}

\aliascntresetthe{proposition}

\newaliascnt{lemma}{theorem}
\newtheorem{lemma}[lemma]{Lemma}
\aliascntresetthe{lemma}
\newaliascnt{corollary}{theorem}

\aliascntresetthe{corollary}

\theoremstyle{definition}
\newaliascnt{definition}{theorem}

\aliascntresetthe{definition}

\newtheorem{algorithm}{Algorithm}

\newaliascnt{remark}{theorem}
\newtheorem{remark}[remark]{Remark}
\aliascntresetthe{remark}

\newaliascnt{example}{theorem}

\aliascntresetthe{example}

\def\rmd{\mathrm{d}}

\def\1{\mathbbm{1}}

